\documentclass[a4paper,reqno,UKenglish]{amsart}
\pdfoutput=1

\usepackage{enumitem}
\usepackage[english]{babel}
\usepackage{xcolor}
\usepackage{amsmath,amssymb,amsthm}
\usepackage{tikz}\usetikzlibrary{arrows}
\usepackage{tikz-cd}
\usepackage{stmaryrd}
\usepackage{extarrows}
\usepackage{pbox}
\usepackage{mathtools}

\usepackage[breaklinks=true]{hyperref}
\definecolor{CustomGreen}{rgb}{0.00,0.50,0.00}
\hypersetup{
    colorlinks = true,
    urlcolor = CustomGreen,
    linkcolor= CustomGreen,
    citecolor= CustomGreen,
    filecolor= CustomGreen,
}

\tikzcdset{
    relay arrow/.default = 10pt,
    relay arrow/.style = {
        rounded corners,
        to path = {
            \pgfextra{
                \def\sourcecoordinate{\pgfpointanchor{\tikztostart}{center}}
                \def\targetcoordinate{\pgfpointanchor{\tikztotarget}{center}}
                \pgfmathanglebetweenpoints{\sourcecoordinate}{\targetcoordinate}
                \edef\tempangle{\pgfmathresult}
                \pgftransformrotate{\tempangle}
                \pgfmathifthenelse{#1>0}{\tempangle+90}{\tempangle-90}
                \pgfcoordinate{tempcoord}{\pgfpointanchor{\tikztostart}{\pgfmathresult}}
            }
            (tempcoord)
            -- ([yshift=#1]tempcoord)
            -- ([yshift=#1]tempcoord-|\tikztotarget.center)\tikztonodes
            --(\tikztotarget)
        }
    }
}

\makeatletter
\def\Ddots{\mathinner{\mkern1mu\raise\p@
\vbox{\kern7\p@\hbox{.}}\mkern2mu
\raise4\p@\hbox{.}\mkern2mu\raise7\p@\hbox{.}\mkern1mu}}
\makeatother

\providecommand{\noopsort}[1]{}

\theoremstyle{plain}
\newtheorem{theorem}{Theorem}
\newtheorem{corollary}[theorem]{Corollary}
\newtheorem{lemma}[theorem]{Lemma}
\newtheorem{proposition}[theorem]{Proposition}

\theoremstyle{definition}
\newtheorem{definition}[theorem]{Definition}
\newtheorem{remark}[theorem]{Remark}

\newtheorem*{claim*}{Claim}
\newtheorem*{fact*}{Fact}

\newcommand{\fin}{\mathbf{FinSet}}

\newcommand{\kri}{\mathsf{Kri}}
\newcommand{\rkri}{\mathsf{RKri}}
\newcommand{\Si}{\Sigma}
\newcommand{\EM}{\mathrm{EM}}
\newcommand{\A}{\mathcal A}
\newcommand{\B}{\mathcal B}
\newcommand{\Tkn}{\mathcal{T}_{k,n}}

\renewcommand{\epsilon}{\varepsilon}
\renewcommand{\phi}{\varphi}
\newcommand{\N}{\mathbb{N}}
\newcommand{\G}{\mathcal{G}} % Gaifman graph
\renewcommand{\k}{\mathbf{k}} % {1,...,k}

\newcommand{\op}{{\mathrm{op}}}
\newcommand{\id}{{\rm id}}
\renewcommand{\tilde}{\widetilde}

\renewcommand{\o}{\overline}
\newcommand{\pit}{\pitchfork}
\newcommand{\im}{\mathrm{im}}

\newcommand\sees{\bowtie}
\newcommand{\up}{{\,\uparrow\,}}
\newcommand{\down}{{\downarrow\,}}
\newcommand{\adj}{{\,\frown\,}}
\newcommand{\tr}[1]{\llbracket #1 \rrbracket}

\newcommand{\St}[2]{{#1}\llbracket #2 \rrbracket}
\newcommand{\Dg}[2]{{#1}\llparenthesis #2 \rrparenthesis}

\newcommand{\mono}{\rightarrowtail}
\newcommand{\epi}{\twoheadrightarrow}
\newcommand{\into}{\hookrightarrow}
\newcommand{\M}{\mathcal{M}}
\newcommand{\E}{\mathcal{E}}

\newcommand{\CC}{\mathbb C}
\newcommand{\Pkn}{\mathbb{P}_{k,n}}
\newcommand{\Pk}{\mathbb{P}_k}
\newcommand{\En}{\mathbb{E}_n}
\newcommand{\Mk}{\mathbb{M}_k}

\newcommand{\CL}{\mathcal C}

\newcommand{\CMLk}{\mathcal C^{\mathrm{ML}}_k}

\newcommand\qtq[1]{\quad\text{#1}\quad}
\newcommand\ee[1]{\enspace{#1}\enspace}

\newcommand{\ef}{Ehrenfeucht--Fra{\"i}ss{\'e}}

\begin{document}
\title{Lov{\'a}sz-Type Theorems and Game Comonads}

\author{Anuj Dawar \and Tom\'a{\v s} Jakl}
\address{Department of Computer Science and Technology\\
University of Cambridge, UK}
\email{anuj.dawar@cl.cam.ac.uk, tj330@cam.ac.uk}

\author{Luca Reggio}
\address{Department of Computer Science\\
University of Oxford, UK}
\email{luca.reggio@cs.ox.ac.uk}
\thanks{
    Research supported in part by  the EPSRC grant EP/T007257/1, and by the European Union's Horizon 2020 research and innovation programme under the Marie Sk{\l}odowska-Curie grant agreement No 837724.
    \\
    \indent
    To appear in the Proceedings of the 36th Annual ACM/IEEE Symposium on Logic in Computer Science (LICS 2021)\relax
}


\maketitle

\begin{abstract}
    Lov\'{a}sz (1967) showed that two finite relational structures $A$ and $B$ are isomorphic if, and only if, the number of homomorphisms from $C$ to $A$ is the same as the number of homomorphisms from $C$ to $B$ for any finite structure $C$. Soon after, Pultr (1973) proved a categorical generalisation of this fact. We propose a new categorical formulation, which applies to any locally finite category with pushouts and a proper factorisation system. As special cases of this general theorem, we obtain two variants of Lov\'{a}sz' theorem: the result by Dvo{\v{r}}{\'a}k (2010) that characterises equivalence of graphs in the $k$-dimensional Weisfeiler--Leman equivalence by homomorphism counts from graphs of tree-width at most $k$, and the result of Grohe (2020) characterising equivalence with respect to first-order logic with counting and quantifier depth $k$ in terms of homomorphism counts from graphs of tree-depth at most $k$.  The connection of our categorical formulation with these results is obtained by means of the game comonads of Abramsky et al.  We also present a novel application to homomorphism counts in modal logic.
\end{abstract}

\maketitle

\section{Introduction}\label{s:intro}

Over fifty years ago, Lov\'{a}sz~\cite{Lovasz1967} proved that two finite graphs (or, more generally, any two finite relational structures) $A$ and $B$ are isomorphic if, and only if, for every finite graph $C$, the number of homomorphisms from $C$ to $A$ is the same as the number of homomorphisms from $C$ to $B$.  While one direction of this equivalence is obvious, the other direction establishes that the isomorphism type of a finite graph $A$ is characterised by an infinite integer vector $V$ indexed by $\mathcal{F}$, the collection of (isomorphism classes of) finite graphs, where for $C \in \mathcal{F}$, $V_C$ is the number of homomorphisms from $C$ to $A$.  This seminal result has led to extensions and investigations in many different directions.  One that concerns us is that by restricting the collection $\mathcal{F}$ to a natural sub-collection, we can often obtain characterisations of natural coarsenings of the relation of isomorphism.  For any subclass $\mathcal{F}'$ of $\mathcal{F}$, say that a pair of graphs $G$ and $H$ is $\mathcal{F}'$-homomorphism equivalent if for any $K$ in $\mathcal{F}'$ the number of homomorphisms from $K$ to $G$ is the same as the number of homomorphisms from $K$ to~$H$. 

For instance, Dvo{\v{r}}{\'a}k~\cite{dvovrak2010recognizing} shows that for any $k$, if we consider the collection $\mathcal{T}_k$ of  all graphs of tree-width at most $k$, then a pair of graphs are $\mathcal{T}_k$-homomorphism equivalent if, and only if, they are equivalent with respect to the $k$-dimensional Weisfeiler--Leman ($k$-WL) equivalence.  The \mbox{$k$-WL} equivalence relations are a widely studied family of approximations of the graph isomorphism relation (see~\cite{Kiefer2020WL} for a recent exposition) with many equivalent characterisations in terms of combinatorics, logic, algebra and linear optimisation.  From our point of view, one important such characterisation is the fact that two graphs $G$ and $H$ are $k$-WL equivalent if, and only if, they are not distinguished by any sentence of $\CL^{k+1}$---first-order logic with counting quantifiers restricted to no more than $k+1$ distinct variables~\cite{CFI92}.  This logical formulation of the equivalence extends naturally to finite structures over any relational vocabulary.

Dvo{\v{r}}{\'a}k's result offers yet another characterisation of $k$-WL equivalence, this time in terms of homomorphism counts from the equally widely studied class of graphs $\mathcal{T}_k$.  An important special case is that of $k=1$.   
Since $\mathcal{T}_1$ is the class of finite forests, and $1$-WL equivalence is the same as \emph{fractional isomorphism} (see~\cite{Ramana1994fractional}), the result also provides an elegant characterisation of equivalence of graphs with respect to the number of homomorphisms from trees.  It is also known that two graphs are co-spectral if, and only if, they are $\mathrm{Cyc}$-homomorphism equivalent, where $\mathrm{Cyc}$ is the class of simple cycles, and that two graphs are \emph{quantum isomorphic} if, and only if, they are $\mathcal{P}$-homomorphism equivalent for the class $\mathcal{P}$ of planar graphs~\cite{manvcinska2020quantum}.  Another recent addition to this collection of results is that of Grohe~\cite{grohe2020counting}.  He shows that, if $\mathcal{D}_n$ is the collection of graphs of \emph{tree-depth} at most $n$, then two graphs are $\mathcal{D}_n$-homomorphism equivalent if, and only if, they are not distinguished by any sentence of $\CL_n$---first-order logic with counting with quantifier depth at most $n$.

This variety of results relating homomorphism counts over restricted classes to approximations of isomorphism, often proved with very different methods, calls for the development of a more general theory.  In the present paper we develop a general categorical result which yields, as example applications, the results of Lov\'{a}sz, Dvo{\v{r}}{\'a}k and Grohe mentioned above.  The connection between our categorical result and the results of Dvo{\v{r}}{\'a}k and Grohe is established using the game comonads of Abramsky et al.~\cite{Abramsky2017b,AbramskyShah2018}.  Specifically, in~\cite{Abramsky2017b} the graded pebbling comonad $\Pk$ on the category of $\sigma$-structures (for a finite relational signature $\sigma$) is introduced and it is shown that the coalgebras for this comonad correspond in a natural way with certain tree decompositions of width at most $k-1$.  In particular, a finite $\sigma$-structure admits a coalgebra structure for this comonad if, and only if, it has tree-width at most $k-1$.  At the same time, it is shown that isomorphism in the Kleisli category corresponding to the comonad (or equivalently in the Eilenberg--Moore category) is exactly indistinguishability in $\CL^k$.  This brings together in one categorical construction the essential elements of Dvo{\v{r}}{\'a}k's theorem.  In an exactly analogous fashion, the {\ef} comonad $\En$ of~\cite{AbramskyShah2018} captures, at the level of coalgebras, the structures of tree-depth at most $n$ and yields a notion of isomorphism corresponding to equivalence in $\CL_n$.  These categorical formulations of the combinatorial parameters tree-width and tree-depth on the one hand, and the logical equivalences with respect to $\CL^k$ and $\CL_n$ on the other, are what enable us to give a single result generalising the three theorems of Lov\'{a}sz, Dvo{\v{r}}{\'a}k and Grohe.

Our general result states that in any locally finite category with pushouts and a proper factorisation system, two objects $m$ and $n$ are isomorphic if, and only if, the number of morphisms from $k$ to $m$ is the same as that from $k$ to $n$ for any object $k$.  This is proved in Section~\ref{s:categorical-Lovasz}.  In Section~\ref{s:applications-fmt}, we use the game comonads to derive the theorems of 
Dvo{\v{r}}{\'a}k and Grohe from this general result and also a combination of the two characterising equivalence in $\CL^k_n$---first-order logic with counting with quantifier depth $n$ and $k$ variables.  In Section~\ref{s:applications-modal}, we apply the machinery we have developed to another game comonad: the graded modal comonad $\Mk$.   This gives a new Lov\'{a}sz-style result for pointed Kripke structures, relating homomorphism counts from synchronization trees of bounded height to equivalence in a modal logic with counting.  Finally, in Section~\ref{s:normal-forms} we draw some conclusions from this about certain normal forms for first-order logic with counting.

\vspace{1ex}
\subsection*{Related work}
Other categorical generalisations of Lov\'asz' theorem have been obtained by Pultr~\cite{pultr1973isomorphism} and Isbell~\cite{isbell1991some}. Cf.\ also Lov\'asz' paper~\cite{Lovasz1972}. In modern terms, Pultr works with a finitely well-powered, locally finite category with (extremal epi, mono) factorisations and Isbell requires the category to be locally finite with a special type of factorisation system (which he calls ``bicategory''). Unlike the results in op.\ cit., our Theorem~\ref{th:Lovasz-categories} does not hinge on a combinatorial counting argument, but rather on an application of the inclusion-exclusion principle, and appears to be better suited for applications to game comonads.

\section{Preliminaries on Game Comonads}\label{s:prelim}
In this section, we recall the necessary material on (game) comonads and their coalgebras. 

\subsection{Categories}
We briefly recall some basic notions of category theory. For a more thorough introduction, see e.g.~\cite{AHS1990} or~\cite{MacLane1998}. 

Let $\A$ be a category, and $f\colon A\to B$ and $g\colon A\to C$ two arrows in $\A$. The \emph{pushout} of $f$ and $g$, if it exists, consists of two arrows $h\colon C\to D$ and $i\colon B\to D$ such that $i\circ f=h\circ g$ and the following universal property is satisfied: For any two arrows $i'\colon B\to E$ and $h'\colon C\to E$ with $i'\circ f=h'\circ g$, there is a unique arrow $\xi\colon D\to E$ satisfying $\xi\circ i=i'$ and $\xi\circ h=h'$.
\[\begin{tikzcd}
    A \arrow{r}{f} \arrow{d}[swap]{g} & B \arrow{d}{i} \arrow[bend left = 30, looseness=1]{ddr}{i'} & \\[0.6em]
C \arrow{r}{h} \arrow[bend right = 30, looseness=1,swap]{drr}{h'} & D \arrow[dashed]{dr}{\xi} & \\
& & E
\end{tikzcd}\]
The \emph{pullback} of two arrows $f\colon A\to C$ and $g\colon B\to C$ in $\A$, if it exists, is the pushout of $f$ and $g$ in the opposite category $\A^\op$ obtained by reversing the direction of the arrows in $\A$.

Pushouts and pullbacks are special instances of the concepts of colimits and limits of diagrams, respectively. For example, the pushout of $f\colon A\to B$ and $g\colon A\to C$ coincides with the colimit of the diagram
\[\begin{tikzcd}[column sep=1.5em]
B & A \arrow{l}[swap]{f} \arrow{r}{g} & C.
\end{tikzcd}\]
For arbitrary colimits, we consider diagrams of any shape.

A functor $F\colon \A\to\B$ is \emph{faithful} if, given any two parallel arrows $f,g\colon A\to A'$ in $\A$, $Ff=Fg$ implies $f=g$. Further, $F$ is \emph{full} if, for any two objects $A$ and $A'$ of $\A$, each arrow $FA\to FA'$ is of the form $Ff$ for some $f\colon A\to A'$. A subcategory $\A$ of a category $\B$ is said to be \emph{full} if the inclusion functor $\A\to \B$ is full.

A fundamental notion of category theory is that of an adjunction. Let $F\colon \A \to \B$ and $G\colon \B \to \A$ be any two functors between categories. We say that $F$ is \emph{left adjoint} to $G$ (equivalently, $G$ is \emph{right adjoint} to $F$), and write $F\dashv G$, provided that for each object $B\in \B$ there is a morphism $\epsilon_B\colon FGB\to B$ satisfying the following universal property: For every $A\in \A$ and morphism $f\colon FA\to B$ in $\B$, there exists a unique morphism $g\colon A\to GB$ in $\A$ such that $\epsilon_B\circ Fg=f$. 
\[\begin{tikzcd}
FGB \arrow{r}{\epsilon_B} & B \\
FA \arrow{ur}[swap]{f} \arrow{u}{Fg}
\end{tikzcd}\]

\subsection{Comonads}\label{s:prelim-comonads}
A \emph{comonad} (\emph{in Kleisli form}) on a category $\A$ is given by:
\begin{itemize}
\item an object map $\CC\colon \A\to\A$;
\item a morphism $\epsilon_A\colon \CC(A)\to A$ for every $A\in\A$;
\item a \emph{coextension operation} associating with any morphism $f\colon \CC(A)\to B$ a morphism $f^*\colon \CC(A)\to \CC(B)$.
\end{itemize}
These data must satisfy the following equations for all morphisms $f\colon \CC(A)\to B$ and $g\colon \CC(B)\to C$:
\[
\epsilon_A^*=\id_{\CC(A)}, \ \ \epsilon_B\circ f^*=f, \ \ (g\circ f^*)^*=g^*\circ f^*.
\]
In particular, we can extend $\CC$ to a functor $\A\to\A$ by setting $\CC(f)\coloneqq(f\circ \epsilon_A)^*$ for every morphism $f\colon A\to B$.\footnote{It is easy to see that, setting $\delta_A\coloneqq \id_{\CC(A)}^*$ for every $A\in \A$, the tuple $(\CC,\epsilon,\delta)$ is a comonad in the more traditional sense, where $\epsilon$ is the counit and $\delta$ the comultiplication. In fact, these two formulations are equivalent.}

We recall next some examples of comonads that play an important role in this paper. Let $\sigma$ be an arbitrary relational signature. The category $\Si$ has as objects $\sigma$-structures (denoted $A,B,\dots$) and as morphisms \emph{$\sigma$-homomorphisms} (sometimes simply called \emph{homomorphisms}), i.e.\ functions $f\colon A\to B$ such that, for all relation symbols ${R\in\sigma}$, $f(R^A)\subseteq R^B$ where $R^A$ and $R^B$ are the interpretations of $R$ in $A$ and $B$, respectively.

For each positive integer $n$, the \emph{\ef} comonad $\En$ on $\Si$ (cf.~\cite{AbramskyShah2018}) is defined as follows. For any $A\in \Si$, the universe of $\En(A)$ is the set $A^{\leq n}$ of all non-empty sequences of length at most $n$. Before defining the interpretations of the relation symbols, let us define the map
\[
\epsilon_A\colon A^{\leq n}\to A, \ \ \epsilon_A[a_1,\ldots, a_l]\coloneqq a_l.
\]
With this notation, for each $R\in\sigma$ of arity $j$, we define $R^{\En(A)}$ to be the set of those tuples $(s_1,\ldots,s_j)$ of sequences which are pairwise comparable in the prefix order, and such that $(\epsilon_A(s_1),\dots,\epsilon_A(s_j))\in R^A$. The coextension operation sends a $\sigma$-homomorphism $f\colon \En(A)\to B$ to the $\sigma$-homomorphism 
\[
f^*\colon \En(A)\to \En(B), \ \ f^*[a_1,\dots,a_j]\coloneqq [b_1,\dots,b_j]
\]
where $b_i\coloneqq f[a_1,\dots,a_i]$ for all $1\leq i\leq j$. It is not difficult to see that these data define a comonad on the category $\Si$.

From the viewpoint of model-comparison games, the elements of $A^{\leq n}$ represent the plays in the $\sigma$-structure $A$ of length at most $n$, and a $\sigma$-homomorphism $\En(A)\to B$ is a strategy for Duplicator in the {\ef} game with $n$ rounds, where Spoiler plays always in $A$ and Duplicator responds in $B$. For more details, cf.~\cite{AbramskyShah2018,AbramskyShah-ext}. Next we recall from~\cite{Abramsky2017b} another comonad on $\Si$, which models pebble games.

For each positive integer $k$, set $\k\coloneqq \{1,\dots,k\}$. Given a $\sigma$-structure $A$, we consider the set $(\k\times A)^+$ of all non-empty sequences of elements of $\k\times A$.   We call a pair $(p,a) \in {\k\times A}$ a \emph{move}.  Whenever $[(p_1,a_1),\dots,(p_l,a_l)]$ is a sequence of moves, $p_i$ is called the \emph{pebble index} of the move $(p_i,a_i)$. As with $\En$, define the map
\[
\epsilon_A\colon (\k\times A)^+\to A, \ \ \epsilon_A[(p_1,a_1),\dots,(p_l,a_l)]\coloneqq a_l
\]
sending a play to its last move. We let $\Pk(A)$ be the $\sigma$-structure with universe $(\k\times A)^+$ and such that, for every relation $R\in\sigma$ of arity $j$, its interpretation $R^{\Pk(A)}$ consists of those tuples of sequences $(s_1,\ldots,s_j)$ such that: (i) the $s_i$ are pairwise comparable in the prefix order; (ii) whenever $s_i$ is a prefix of~$s_{i'}$, the pebble index of the last move in $s_i$ does not appear in the suffix of $s_i$ in $s_{i'}$; and (iii) ${(\epsilon_A(s_1),\dots,\epsilon_A(s_j))\in R^A}$. Finally, the coextension operation is the same, mutatis mutandis, as with $\En$. Again, it is not difficult to see that these data define a comonad $\Pk$ on $\Si$, called the \emph{pebbling comonad}.

The {\ef} comonad $\En$ can be ``combined'' with the pebbling comonad $\Pk$ to obtain a new comonad~$\Pkn$ on $\Si$, cf.~\cite{paine2020pebbling}. For every $A\in\Si$, the universe of $\Pkn(A)$ is the set $(\k\times A)^{\leq n}$ of all non-empty sequences of length at most~$n$ (this corresponds to bounding the length of plays in pebble games). The homomorphisms $\epsilon_A$ and the coextension operation for $\Pkn$ are defined as the restrictions of the corresponding concepts for the pebbling comonad $\Pk$.

\subsection{Coalgebras}\label{s:prelim-coalgebras}
Let $\CC$ be a comonad (in Kleisli form) on a category $\A$. A \emph{coalgebra} for $\CC$ is a pair $(A, A\xrightarrow{\alpha} \CC(A))$ such that $A\in \A$, $\alpha$ is a morphism in $\A$, and the following diagrams commute:
\begin{center}
\begin{tikzcd}[column sep=2.5em, row sep=2.5em]
A \arrow{dr}[swap]{\id_A} \arrow{r}{\alpha} & \CC(A) \arrow{d}{\epsilon_A} \\
{} & A
\end{tikzcd}
\ \ \ \ \ \ 
\begin{tikzcd}[row sep=2.5em]
A \arrow{r}{\alpha} \arrow{d}[swap]{\alpha} & \CC(A) \arrow{d}{\id_{\CC(A)}^*} \\
\CC(A) \arrow{r}{\CC(\alpha)} & \CC(\CC(A))
\end{tikzcd}
\end{center}
We refer to $\alpha$ as a \emph{coalgebra structure} for $A$. In the case of the game comonads $\En,\Pk$, and $\Pkn$, the $\sigma$-structures admitting a coalgebra structure can be characterised in an elegant way in terms of key combinatorial parameters, as we now explain. 

Let $(F,\leq)$ be any poset. For all $x,y\in F$, we write $x\up y$ whenever $x$ and $y$ are comparable, i.e.\ either $x\leq y$ or $y\leq x$, and let $\down x\coloneqq \{y\in F\mid y\leq x\}$. A \emph{forest} is a poset $(F,\leq)$ such that, for all $x\in F$, the set $\down x$ is finite and totally ordered. The \emph{height} of $(F,\leq)$ is $\sup_{x\in F}{|\down x|}$.

If $G=(V,\frown)$ is a graph (where $V$ is the set of vertices and $\frown$ the adjacency relation), a \emph{forest cover} of $G$ consists of a forest $(F,\leq)$ and an injective function $f\colon V\to F$ such that $v\adj v'$ entails $f(v)\up f(v')$ for all $v,v'\in V$.  

Finally, recall that the \emph{Gaifman graph} of a $\sigma$-structure $A$ is $\G_A\coloneqq (A,\frown)$ where $a\adj a'$ if $a\neq a'$ and there exist $R\in \sigma$ and a tuple $(a_1,\dots,a_l)\in R^A$ such that $a=a_i$ and $a'=a_j$ for some $i,j\in\{1,\dots,l\}$. A $\sigma$-structure $A$ has \emph{tree-depth at most $n$} if there is a forest cover $(F,\leq)$ of $\G_A$ of height ${\leq \,} n$.

The following is a direct consequence of \cite[Theorem~17]{AbramskyShah2018}:
\begin{proposition}\label{p:En-coalgebras}
A $\sigma$-structure $A$ admits a coalgebra structure $\alpha\colon A\to \En(A)$ if, and only if, it has tree-depth at most $n$.
\end{proposition}

In the same way as $\sigma$-structures that admit a coalgebra structure for $\En$ can be characterised in terms of tree-depth, the existence of a coalgebra structure for the pebbling comonad $\Pk$ corresponds to tree-width. We recall the relevant definitions.

Note that, if $(F,\leq)$ is a forest cover of a graph $G=(V,\frown)$, with injective map $f\colon V\to F$, we can assume without loss of generality that $F=V$ and $f=\id_V$. We shall assume this for the remainder of this section.
A \emph{k-pebble forest cover} of a graph $G=(V,\frown)$ consists of a forest cover $(V,\leq)$ of $G$ and a pebbling function $p\colon V\to \k$ such that, whenever $v\adj v'$ with $v\leq v'$, we have $p(v)\neq p(w)$ for all $v<w\leq v'$.

A $\sigma$-structure $A$ has \emph{tree-width at most $k-1$} if its Gaifman graph $\G_A$ admits a $k$-pebble forest cover.\footnote{It is customary in graph theory to use $k-1$, instead of $k$, so that trees have tree-width $1$.} The following result is a consequence of \cite[Proposition~22]{Abramsky2017b}.
 \begin{proposition}\label{p:Pk-coalgebras}
A $\sigma$-structure $A$ admits a coalgebra structure $\alpha\colon A\to \Pk(A)$ if, and only if, it has tree-width at most $k-1$.
\end{proposition}

\begin{remark}
The tree-width of a graph is usually defined in terms of tree decompositions. The definition given above is an equivalent reformulation. For a proof of the fact that a finite graph admits a tree decomposition of width ${< \,} k$ if, and only if, it admits a $k$-pebble forest cover, see \cite[Theorem~19]{AbramskyShah2018}.
\end{remark}

In the same spirit, in the case of the comonad $\Pkn$ we have the following characterisation, cf.~\cite[Theorem~2.14]{paine2020pebbling}.

\begin{proposition}\label{p:Pkn-coalgebras}
A $\sigma$-structure $A$ admits a coalgebra structure $\alpha\colon A\to \Pkn(A)$ if, and only if, it has a $k$-pebble forest cover of height ${\leq \,} n$.
\end{proposition}

Finally, recall that coalgebras for a comonad $\CC$ on $\A$ form themselves a category $\EM(\CC)$, the \emph{Eilenberg--Moore category} of $\CC$. The objects of $\EM(\CC)$ are the coalgebras $(A,\alpha)$ for $\CC$, and morphisms $(A,\alpha)\to (B,\beta)$ in $\EM(\CC)$ are morphisms $h\colon A\to B$ in $\A$ such that $\CC(h)\circ\alpha=\beta\circ h$. There is an obvious forgetful functor 
\[ U^{\CC}\colon \EM(\CC)\to \A \]
which sends a coalgebra $(A,\alpha)$ to $A$. This functor has a right adjoint
\[ F^{\CC}\colon \A\to \EM(\CC)\]
defined as follows: for any $A\in \A$, $F^{\CC}(A)\coloneqq (\CC(A),\delta_A)$ where $\delta_A\coloneqq \id_{\CC(A)}^*$. Further, if $h$ is a morphism in $\A$, we set $F^{\CC}(h)\coloneqq \CC(h)$. 

\vspace{-0.3em}
\[\begin{tikzcd}[column sep=1.5em]
        \EM(\CC) \arrow[yshift=-0.3pt,out=345,in=205]{rr}[swap]{U^{\CC}}
        & {\text{\scriptsize{$\top$}}}
        & \A \arrow[yshift=0.3pt,out=155,in=15]{ll}[swap]{F^{\CC}}
\end{tikzcd}\]

\section{A Categorical Lov{\'a}sz-Type Theorem}\label{s:categorical-Lovasz}

A category is \emph{locally finite} if there are only finitely many morphisms between any two of its objects. A locally finite category $\A$ is said to be \emph{combinatorial} if, for all $m,n\in\A$,
\[
m\cong n \ \Longleftrightarrow \ |\hom_{\A}(k,m)|=|\hom_{\A}(k,n)| \ \ \forall k\in\A.
\]
Thus, Lov{\'a}sz' theorem \cite{Lovasz1967} states that, for any finite relational signature~$\sigma$, the category $\Si_f$ of finite $\sigma$-structures with homomorphisms is combinatorial.
The aim of this section is to prove the following generalisation of Lov{\'a}sz' result.
\begin{theorem}\label{th:Lovasz-categories}
Let $\A$ be a locally finite category. If $\A$ has pushouts and a proper factorisation system, then it is combinatorial.
\end{theorem}

\subsection{Proper factorisation systems}
We start by recalling the notion of weak factorisation system.
Given morphisms $e$ and $m$ in $\A$, we say that $e$ has the \emph{left lifting property} with respect to $m$, or that $m$ has the \emph{right lifting property} with respect to $e$ if, for every commutative square as on the left-hand side below,
\begin{equation*}
\begin{tikzcd}
{\bullet} \arrow{d} \arrow{r}{e} & {\bullet} \arrow{d} \\
{\bullet} \arrow{r}{m} & {\bullet}
\end{tikzcd}
\ \ \ \ \ \ \ \ \ \ \ \ \ 
\begin{tikzcd}
{\bullet} \arrow{d} \arrow{r}{e} & {\bullet} \arrow{d} \arrow{dl}[swap, outer sep=-1pt]{d} \\
{\bullet} \arrow{r}{m} & {\bullet}
\end{tikzcd}
\end{equation*}
there is a (not necessarily unique) \emph{diagonal filler}, i.e.\ an arrow $d$ such that the right-hand diagram above commutes. If this is the case, we write $e {\,\pit\,} m$. For any class $\mathcal{F}$ of morphisms in $\A$, let ${}^{\pit}\mathcal{F}$ (respectively $\mathcal{F}^{\pit}$) be the class of morphisms having the left (respectively right) lifting property with respect to every morphism in $\mathcal{F}$.

\begin{definition}\label{def:weak-f-s}
A pair of classes of morphisms $(\E,\M)$ in a category $\A$ is a \emph{weak factorisation system} provided it satisfies the following conditions:
\begin{enumerate}
\item every morphism $f$ in $\A$ can be written as $f = m \circ e$ with $e\in \E$ and $m\in \M$;
\item $\E={}^{\pit}\M$ and $\M=\E^{\pit}$.
\end{enumerate}
A \emph{proper factorisation system} is a weak factorisation system $(\E,\M)$ such that all morphisms in $\E$ are epimorphisms and all morphisms in $\M$ are monomorphisms.\footnote{It is an easy observation that any proper factorisation system is an \emph{orthogonal} factorisation system, meaning that the diagonal fillers are unique.}
\end{definition}

We describe two different proper factorisation systems in the category $\Si$ of $\sigma$-structures, for a fixed relational signature~$\sigma$. Each homomorphism of $\sigma$-structures $f\colon A\to B$ factors, as a function, through its set-theoretic image $\tilde{A}$:
\[
A \to \tilde{A} \to B.
\] 
There are two natural ways to turn $\tilde{A}$ into a $\sigma$-structure: we can equip it with the structure induced by either $B$, or $A$. In the first case, for every $R\in\sigma$ of arity $n$, we set $R^{\tilde{A}}\coloneqq R^{B}\cap \tilde{A}^n$ and, in the second, we set $R^{\tilde{A}}\coloneqq f(R^A)\cap \tilde{A}^n$. 

These two ways of turning $\tilde{A}$ into a $\sigma$-structure yield two different weak factorisation systems $(\E,\M)$ in the category $\Si$.\footnote{These are the (epi, regular mono) and (regular epi, mono) factorisation systems, respectively. Cf.~\cite[pp.~200--201]{AR1994}.} 
For example, the first one yields the weak factorisation system where
$\E=\{\text{surjective homomorphisms}\}$ and $\M=\{\text{strong/induced embeddings}\}$.
Both factorisation systems are proper. Moreover, since the image of a finite $\sigma$-structure under a homomorphism is also finite, these factorisation systems restrict to the locally finite category $\Si_f$ of finite $\sigma$-structures.

The category $\Si$ also has pushouts. Given two homomorphisms of $\sigma$-structures $f\colon A\to B$ and $g\colon A\to C$, their pushout $D$ is computed as the quotient of the disjoint sum of $\sigma$-structures $B+C$ by the least equivalence relation $\sim$ such that $f(a) \sim g(a)$ for every $a\in A$.
Equivalently, $D$ is obtained by equipping $D'$, the pushout of the functions $f$ and $g$ in the category of sets, depicted in the following diagram, 
\[
    \begin{tikzcd}
        A \rar{f}\dar[swap]{g} & B \dar[dashed]{f'} \\
        C \rar[dashed]{g'} & D' 
    \end{tikzcd}
\]
with the smallest relational structure that turns $f'$ and $g'$ into homomorphisms.
  Using the same idea, it is not difficult to see that $\Si$ is cocomplete, i.e.\ it has all colimits, cf.~\cite[p.~201]{AR1994}.

By the previous description of pushouts, it is clear that the pushout in $\Si$ of finite $\sigma$-structures is again a finite $\sigma$-structure. It follows that $\Si_f$ has all pushouts, and thus it satisfies the assumptions of Theorem~\ref{th:Lovasz-categories}. This justifies our claim that the latter result is a generalisation of Lov{\'a}sz' theorem.

\vspace{1ex}
We state next some well known properties of weak and proper factorisation systems (cf., e.g.,~\cite{freyd1972categories} or~\cite{riehl2008factorization}):
\begin{lemma}\label{l:factorisation-properties}
Let $(\E,\M)$ be a weak factorisation system in $\A$. The following hold:
\begin{enumerate}[label=\textup{(}\textbf{\alph*}\textup{)}]
\item\label{compositions} $\E$ and $\M$ are closed under compositions;
\item\label{isos} $\E\cap\M=\{\text{isomorphisms}\}$;
\item\label{pushouts} the pushout in $\A$ of an $\E$-morphism along any morphism, if it exists, is again in $\E$.
\end{enumerate}
Moreover, if $(\E,\M)$ is proper, the following hold:
\begin{enumerate}[label=\textup{(}\textbf{\alph*}\textup{)}]\setcounter{enumi}{3}
\item\label{cancellation-e} $g\circ f\in \E$ implies $g\in\E$;
\item\label{cancellation-m} $g\circ f\in\M$ implies $f\in\M$.
\end{enumerate}
\end{lemma}

\subsection{Proof of Theorem~\ref{th:Lovasz-categories}}
For the remainder of this section, we fix a category $\A$ admitting a proper factorisation system $(\E,\M)$. $\M$-morphisms will be denoted by $\mono$, and $\E$-morphisms by $\epi$.

An \emph{$\E$-pushout square} in $\A$ is a pushout square consisting entirely of $\E$-morphisms. We are interested in functors $\A^\op\to\fin$, where $\fin$ is the category of finite sets and functions between them, that turn $\E$-pushout squares in $\A$ into pullbacks in $\fin$. Recall that pullbacks in $\fin$ admit the following explicit description: a commutative square in $\fin$ is a pullback if, and only if,
\[\begin{tikzcd}
A \arrow{r}{f} \arrow{d}[swap]{g} & B \arrow{d}{i} \\
C \arrow{r}{h} & D
\end{tikzcd} 
\ \ \ \ \
\pbox{18em}{\relax\ifvmode\centering\fi
$\forall b\in B, \forall c\in C, \  \text{if} \ i(b)=h(c) \ \text{then}$\\ 
$\exists! a\in A. \ (f(a)=b \text{ and } g(a)=c).$} 
\]
If $h$ and $i$ are injective then, by identifying $B$ and $C$ with subsets of $D$, the pullback $A$ can be identified with $B\cap C$.

The following lemma shows that the hom-functors involved in the homomorphism counting theorem send pushout squares to pullbacks. (This is a consequence of the more general fact that representable functors preserve all limits that exist in their domain, cf.~\cite[pp.~116-117]{MacLane1998}, but we provide here also a direct elementary proof.) Recall that, for any object $n$ of a locally finite category $\B$, the functor $\hom_{\B}(-,n)\colon \B^\op\to\fin$ sends a morphism $f\colon a\to b$ in $\B$ to the function $-\circ f\colon \hom_{\B}(b,n)\to \hom_{\B}(a,n)$.
\begin{lemma}\label{l:hom-lf-limits}
For any locally finite category $\B$ and $n\in\B$, the functor
\[ \hom_{\B}(-,n)\colon \B^\op\to\fin\]
sends all pushout squares that exist in $\B$ to pullbacks in $\fin$.
\end{lemma}
\begin{proof}
Just observe that, given a pushout square in $\B$ and the corresponding diagram in $\fin$, as displayed below,
\begin{center}
\begin{tikzcd}
a \arrow{d}[swap]{g} \arrow{r}{f} & b \arrow{d}{i} \\
c \arrow{r}{h} & d
\end{tikzcd}
\ \ \ \ \ \ \ \ \ 
\begin{tikzcd}
\hom_{\B}(d,n) \arrow{r}{-\circ i} \arrow{d}[swap]{-\circ h} & \hom_{\B}(b,n) \arrow{d}{-\circ f} \\
\hom_{\B}(c,n) \arrow{r}{-\circ g} & \hom_{\B}(a,n)
\end{tikzcd}
\end{center}
if $\alpha\in \hom_{\B}(b,n)$ and $\beta\in\hom_{\B}(c,n)$ are such that $\alpha\circ f=\beta\circ g$ then, by the universal property of the pushout, there is a unique $\gamma\in \hom_{\B}(d,n)$ satisfying $\gamma\circ i=\alpha$ and $\gamma\circ h=\beta$.
\end{proof}

\begin{lemma}\label{l:E-morph-to-injections}
If a functor $F\colon \A^\op\to \fin$ sends $\E$-pushout squares in $\A$ to pullbacks, then it sends $\E$-morphisms to injections.
\end{lemma}
\begin{proof}
Let $e\colon n\epi m$ be an $\E$-morphism in the category $\A$. Because $e$ is an epimorphism, it follows directly that the square on the left-hand side below is a pushout square in $\A$.
\begin{center}
\begin{tikzcd}
n \arrow[twoheadrightarrow]{d}[swap]{e} \arrow[twoheadrightarrow]{r}{e} & m \arrow{d}{\id} \\
m \arrow{r}{\id} & m
\end{tikzcd}
\ \ \ \ \ \ \ \ \ 
\begin{tikzcd}
F(m) \arrow{r}{\id} \arrow{d}[swap]{\id} & F(m) \arrow{d}{F(e)}  \\
F(m) \arrow{r}{F(e)} & F(n)  
\end{tikzcd}
\end{center}
Since identities are $\E$-morphisms, the square on the right-hand side above is a pullback in $\fin$. In turn, this is equivalent to $F(e)$ being an injection.
\end{proof}

\subsubsection*{Generic and degenerate elements}
The key step in the proof of the classical Lov\'asz' theorem consists in showing that, if $|\hom_{\Si_f}(C,A)|=|\hom_{\Si_f}(C,B)|$ for all $C$ in $\Si_f$, then also $|\operatorname{inj}_{\Si_f}(C,A)|=|\operatorname{inj}_{\Si_f}(C,B)|$ for all $C$ in $\Si_f$, where $\operatorname{inj}_{\Si_f}(C,A)$ is the set of all injective homomorphisms $C \hookrightarrow A$, and similarly for $\operatorname{inj}_{\Si_f}(C,B)$. We aim to show the analogous fact in our setting, where injective homomorphisms are replaced by $\M$-morphisms.

The intuition behind the following definition is that a $\sigma$-homomorphism $f\colon C\to A$ is non-injective precisely when there exist an onto non-injective $\sigma$-homomorphism $e\colon C \epi C'$ and a $\sigma$-homomorphism $g\colon C' \to A$ such that $f = g \circ e$. In terms of the hom-functor $E \coloneqq \hom_{\Si_f}(-,A)$, $f\in E(C)$ is non-injective if, and only if, $f = E(e)(g)$ for some onto non-injective $\sigma$-homomorphism $e\colon C \epi C'$ and $g\in E(C')$.

\begin{definition}
A \emph{strict quotient} in $\A$ is an $\E$-morphism which is not an isomorphism. (Equivalently, by Lemma~\ref{l:factorisation-properties}\ref{isos}, a strict quotient is an arrow in $\E\setminus \M$.)

Further, consider a functor $E\colon \A^\op\to\fin$. For every $k\in\A$, we say that $s\in E(k)$ is \emph{degenerate} if there exist a strict quotient $f\colon k\epi l$ and $t\in E(l)$ such that ${E(f)(t)=s}$. Otherwise, $s$ is called \emph{generic}. The subset of $E(k)$ consisting of the generic elements is denoted by $\St{E}{k}$.
\end{definition}

The next lemma shows that this definition matches our intuition for hom-functors on $\A$.
\begin{lemma}\label{l:generic-of-hom}
    Let $E = \hom_\A(-,n)$ for some $n\in \A$. For any $k\in\A$, $\St{E}{k}$ is the set of all $\M$-morphisms $k\to n$.
\end{lemma}
\begin{proof}
    Let $f$ be an arbitrary element of $E(k)=\hom_{\A}(k,n)$. Assume $f$ is generic, and take its $(\E,\M)$ factorisation:
    \[\begin{tikzcd}
        k \arrow[rr, relay arrow=2ex, "f"] \arrow[r, twoheadrightarrow, "g"] & l \arrow[r, rightarrowtail, "h"] & n
    \end{tikzcd}\]
    Then $E(g)(h)=f$, and so $g$ must be an $\M$-morphism. By Lemma~\ref{l:factorisation-properties}\ref{compositions}, $f=h\circ g$ is also an $\M$-morphism. 

    Conversely, suppose $f$ is an $\M$-morphism and pick an $\E$-morphism $g\colon k\epi l$ and $h\in E(l)$ such that $E(g)(h)=f$, i.e.\ $h\circ g=f$. By Lemma~\ref{l:factorisation-properties}\ref{cancellation-m}, $g$ is an $\M$-morphism, so it is not a strict quotient. That is, $f$ is generic.
\end{proof}

The following is the main technical lemma of this section, and it ultimately relies on an application of the inclusion-exclusion principle.
\begin{lemma}\label{l:cores-isomorphic}
Assume $\A$ has pushouts, and let $E$ and $F$ be functors $\A^\op\to \fin$ sending $\E$-pushout squares in $\A$ to pullbacks. If $E(k)\cong F(k)$ for all $k\in \A$, then $\St{E}{k}\cong \St{F}{k}$ for all $k\in\A$.
\end{lemma}
\begin{proof}
Let $E,F$ be as in the statement, and suppose that $E(k)\cong F(k)$ for all $k\in \A$. We must show that, for every $k\in \A$, $\St{E}{k}\cong \St{F}{k}$. Equivalently, we show that the set $\Dg{E}{k}$ of degenerate elements of $E(k)$ is in bijection with the set $\Dg{F}{k}$ of degenerate elements of $F(k)$. 

Observe that $\Dg{E}{k}$ coincides with
\[
\bigcup{\{\operatorname{Im}(E(f))\mid f\colon k \epi l \ \text{is a strict quotient in $\A$}\}}\subseteq E(k),
\]
and similarly for $\Dg{F}{k}$.
    Since $E(k)$ and $F(k)$ are finite, there are finite sets $\mathcal{S}_1, \mathcal{S}_2$ of strict quotients of $k$  such that
\begin{align*} 
\Dg{E}{k}&=\bigcup{\{\operatorname{Im}(E(f))\mid f\in\mathcal{S}_1\}}, \\
\Dg{F}{k}&=\bigcup{\{\operatorname{Im}(F(f))\mid f\in\mathcal{S}_2\}}.
\end{align*}
Let $\mathcal{S}\coloneqq\mathcal{S}_1\cup \mathcal{S}_2$. By the inclusion-exclusion principle,
    \[ |\Dg{E}{k}| = \sum_{J\subseteq \mathcal S, J\neq \emptyset} (-1)^{|J|+1}\enspace \big|\!\bigcap_{f\in J} \operatorname{Im}(E(f))\big| \]
    and similarly for $|\Dg{F}{k}|$. Therefore, to prove $\Dg{E}{k}\cong \Dg{F}{k}$ it suffices to show that $\big|\!\bigcap_{f\in J} \operatorname{Im}(E(f))\big|=\big|\!\bigcap_{f\in J} \operatorname{Im}(F(f))\big|$ for any non-empty $J\subseteq \mathcal S$. 
    To this end, fix such a $J$ and consider the \emph{wide pushout} in $\A$ of the strict quotients in $J$, i.e.\ the diagram obtained by taking consecutive pushouts of the elements of $J$, as shown in the diagram on the left-hand side below. 
 \begin{center}
 \begin{tikzcd}[column sep=0.4em, row sep=0.3em]
 k \arrow[twoheadrightarrow]{ddr} \arrow[twoheadrightarrow]{drr} \arrow[twoheadrightarrow]{rrr} \arrow[twoheadrightarrow]{ddd} & & & l_u \arrow[twoheadrightarrow]{ddd} \\
 {} &   & l_3 \arrow[twoheadrightarrow]{ddr} \arrow[ur, phantom, "\Ddots"] & \\ 
 {} & l_2 \arrow[twoheadrightarrow]{drr} &   & \\ 
 l_1 \arrow[twoheadrightarrow]{rrr} & & & p
 \end{tikzcd}
 \hspace{4em}
  \begin{tikzcd}[column sep=-1em, row sep=0.3em]
 E(k) \arrow[hookleftarrow]{ddr} \arrow[hookleftarrow]{drr} \arrow[hookleftarrow]{rrr} \arrow[hookleftarrow]{ddd} & & & E(l_u) \arrow[hookleftarrow]{ddd} \\
 {} &   & E(l_3) \arrow[hookleftarrow]{ddr} \arrow[ur, phantom, "\Ddots"] & \\
 {} & E(l_2) \arrow[hookleftarrow]{drr} &   & \\
 E(l_1) \arrow[hookleftarrow]{rrr} & & & E(p)
 \end{tikzcd}
 \end{center}
By Lemma~\ref{l:factorisation-properties}\ref{pushouts}, all arrows in this diagram are $\E$-morphisms. Hence, by Lemma~\ref{l:E-morph-to-injections}, the diagram obtained by applying the functor $E$, depicted on the right-hand side above, consists of injections. Since $E$ sends $\E$-pushouts to pullbacks, the diagram on the right is a wide pullback of injections in $\fin$, and so $E(p) \cong \bigcap_{f\in J} \operatorname{Im}(E(f))$. Similarly,  $F(p) \cong \bigcap_{f\in J} \operatorname{Im}(F(f))$. 
As $E(p)\cong F(p)$ by assumption, we get $\big|\!\bigcap_{f\in J} \operatorname{Im}(E(f))\big|=\big|\!\bigcap_{f\in J} \operatorname{Im}(F(f))\big|$, thus concluding the proof.
\end{proof}

We are now in a position to prove the main result of this section:
\begin{proof}[Proof of Theorem~\ref{th:Lovasz-categories}]
Fix arbitrary objects $m,n\in\A$. 
For the non-trivial direction, assume $|\hom_{\A}(k,m)| = |\hom_{\A}(k,n)|$ for all $k\in\A$.
Let $E$ and $F$ denote, respectively, the functors $\hom_{\A}(-,m)\colon \A^\op\to\fin$ and $\hom_{\A}(-,n)\colon \A^\op\to\fin$. These functors send $\E$-pushout squares in $\A$ to pullbacks by Lemma~\ref{l:hom-lf-limits}. Thus, by Lemma~\ref{l:cores-isomorphic}, $\St{E}{k}\cong \St{F}{k}$ for all $k\in \A$. According to Lemma~\ref{l:generic-of-hom}, there is a bijection
\begin{equation*}
\{\text{$\M$-morphisms} \ k\mono m\} \ee\cong \{\text{$\M$-morphisms} \ k\mono n\}.
\end{equation*}
    Setting $k=m$, the existence of an $\M$-morphism $m\mono m$ (namely the identity) entails the existence of an $\M$-morphism $i\colon m\mono n$. Similarly, there exists an $\M$-morphism $j\colon n\mono m$. A standard argument then shows that $m \cong n$ (see e.g.~\cite{pultr1973isomorphism}). We briefly sketch a proof. The set $L$ of all $\M$-morphisms $m\mono m$ is a monoid with respect to composition and contains $j\circ i$. Since all $\M$-morphisms are monos, $L$ satisfies the left cancellation law $a b=a c \Rightarrow b=c$. But every finite monoid satisfying the left cancellation law is a group, hence $j\circ i$ has an inverse. It follows from Lemma~\ref{l:factorisation-properties}\ref{isos},\ref{cancellation-e} that $j$ is an isomorphism.
\end{proof}

\begin{remark}
Direct inspection of the preceding proofs shows that Theorem~\ref{th:Lovasz-categories} can be slightly strengthened by requiring only the existence of pushouts of $\E$-morphisms along $\E$-morphisms.
\end{remark}

\subsection{Specialising to finite coalgebras}
To conclude this section, we show that Theorem~\ref{th:Lovasz-categories} can be used to lift Lov{\'a}sz' homomorphism counting result from $\Si_f$ to categories of finite coalgebras for comonads on $\Si$. This result is then applied in Sections~\ref{s:applications-fmt}--\ref{s:applications-modal} to obtain homomorphism counting results in finite model theory and modal logic, respectively.

Given a comonad $\CC$ on $\Si$, we say that a coalgebra $(A,\alpha)$ for $\CC$ is \emph{finite} if $A$ is a finite $\sigma$-structure. The full subcategory of $\EM(\CC)$ defined by the finite coalgebras is denoted $\EM_f(\CC)$.
\begin{corollary}\label{cor:combinatorial-EM-f}
Let $\CC$ be any comonad on $\Si$. Then $\EM_f(\CC)$, the category of finite coalgebras for $\CC$, is combinatorial.
\end{corollary}

In order to prove the previous result, we need the following notions.
Recall that, for any category $\A$ and object $A\in\A$, two epimorphisms $f\colon A\to B$ and $g\colon A\to C$ are \emph{equivalent}, written $f\sim g$, if there exists an isomorphism $h\colon B\to C$ such that $h\circ f=g$. It is not difficult to see that $\sim$ is an equivalence relation on the collection of epimorphisms with domain $A$. The category $\A$ is said to be \emph{well-copowered} if, for each $A\in\A$, the collection of all $\sim$-equivalence classes of epimorphisms with domain $A$ is a set (as opposed to a proper class). Any cocomplete category that is well-copowered admits a proper factorisation system, see e.g.~\cite[Propositions~4.4.2,~4.4.3]{Borceux1} and \cite[Proposition~14.7]{AHS1990}.

Moreover, recall that a functor $F\colon \A \to \B$ \emph{creates colimits} if, for every diagram $\mathcal D$ in $\A$ such that $F(\mathcal D)$ admits a colimit $B$ in $\B$, there exists a colimit $A$ of $\mathcal D$ in $\A$ with $F(A)\cong B$. In particular, if $\B$ is cocomplete and $F$ creates colimits, then $\A$ is also cocomplete and the functor $F$ preserves colimits. Further, $F$ \emph{creates isomorphisms} if, whenever $h\colon B \to F(A)$ is an isomorphism in $\B$, there exists an isomorphism $h'\colon A'\to A$ in $\A$ such that $F(h') = h$. It is a well known fact that, for any comonad $\CC$ on $\A$, the forgetful functor $U^{\CC}\colon \EM(\CC)\to \A$ creates colimits and isomorphisms (see e.g.\ \cite[Proposition~20.12]{AHS1990} for a proof of the dual statement).

\begin{proof}[Proof of Corollary~\ref{cor:combinatorial-EM-f}]
We show that the category $\EM_f(\CC)$ satisfies the hypotheses of Theorem~\ref{th:Lovasz-categories}. The forgetful functor $U^{\CC}\colon \EM(\CC)\to \Si$ is faithful and restricts to a functor $\EM_f(\CC)\to \Si_f$. Thus, for all $(A,\alpha),(B,\beta)\in\EM_f(\CC)$, there is an injection $\hom_{\EM_f(\CC)}((A,\alpha),(B,\beta))\into \hom_{\Si_f}(A, B)$. Because the latter set is finite, so is the former. Therefore, $\EM_f(\CC)$ is locally finite.

We prove next that $\EM_f(\CC)$ has all pushouts. Because the inclusion $\EM_f(\CC)\into \EM(\CC)$ is full and faithful, it suffices to prove that the pushout in $\EM(\CC)$ of finite coalgebras exists and is a finite coalgebra. 
Since $\Si$ is cocomplete (cf.\ the discussion following Definition~\ref{def:weak-f-s}) and $U^{\CC}\colon \EM(\CC)\to \Si$ creates colimits, the category $\EM(\CC)$ is cocomplete and the forgetful functor $U^{\CC}\colon \EM(\CC)\to \Si$ preserves all colimits. In particular, the pushout of finite coalgebras exists in $\EM(\CC)$ and is a finite coalgebra. 

The functor $U^{\CC}\colon \EM(\CC)\to\Si$ preserves colimits, and in particular epimorphisms. Using the fact that $U^{\CC}$ is faithful and creates isomorphisms, it is not difficult to see that $\EM(\CC)$ is well-copowered because so is $\Si$. Thus, $\EM(\CC)$ admits a proper factorisation system $(\E,\M)$.
    Now, consider a morphism $(A,\alpha)\to (B,\beta)$ in $\EM_f(\CC)$ and its $(\E,\M)$-factorisation in $\EM(\CC)$:
\[\begin{tikzcd}
(A,\alpha) \arrow{r}{e} & (\tilde{A},\tilde{\alpha}) \arrow{r}{m} & (B,\beta)
\end{tikzcd}\]
Since $U^{\CC}$ preserves epimorphisms, $U(e)\colon A\to \tilde{A}$ is an epimorphism in $\Si$, that is a surjective homomorphism. It follows that $\tilde{A}$ is a finite $\sigma$-structure and so $(\tilde{A},\tilde{\alpha})\in\EM_f(\CC)$.

The fact that $\EM_f(\CC)$ admits a proper factorisation system is then a consequence of the following easy observation: Let $\A$ be any category equipped with a proper factorisation system $(\E,\M)$, and $\B$ a full subcategory of $\A$. Assume that, whenever $A\to B$ is a morphism in $\B$ and $A\to \tilde{A}\to B$ is its $(\E,\M)$-factorisation in $\A$, also $\tilde{A}\in \B$. Then $(\E\cap \B,\M\cap \B)$ is a proper factorisation system in $\B$.
\end{proof}

\begin{remark}\label{r:pointed-structures}
If $\CC$ is a comonad on $\Si$ that restricts to a comonad $\CC'$ on $\Si_f$, then $\EM(\CC')$ is isomorphic to $\EM_f(\CC)$ and thus $\EM(\CC')$ is combinatorial by Corollary~\ref{cor:combinatorial-EM-f}.
Also, note that the same proof as for Corollary~\ref{cor:combinatorial-EM-f} applies, mutatis mutandis, if $\Si$ is replaced by the category of \emph{pointed $\sigma$-structures}, having as objects the pairs $(A,a)$ with $A\in\Si$ and $a\in A$, and as morphisms the $\sigma$-homomorphisms preserving the distinguished points. This will be needed in Section~\ref{s:applications-modal} for applications in the setting of modal logic.

More generally, it is not difficult to see that we may replace $\Si$ and $\Si_f$ by any category $\A$, and full subcategory $\A_f$ of $\A$, respectively, such that:
    \begin{enumerate}
        \item $\A$ is cocomplete and well-copowered;
        \item $\A_f$ is locally finite, closed under finite colimits in $\A$ and, if $A \to B$ is an epimorphism in $\A$ with $A\in \A_f$, then $B\in \A_f$.
    \end{enumerate}
\end{remark}

\section{Applications to Finite Model Theory}\label{s:applications-fmt}
In this section, we apply Corollary~\ref{cor:combinatorial-EM-f} to the game comonads introduced in Section~\ref{s:prelim} to give new proofs of results of Grohe and Dvo\v r\'ak connecting homomorphism counts and indistinguishability in appropriate logic fragments.
We start by establishing a general result (Theorem~\ref{t:general-statement}) and then proceed to show how Grohe's and Dvo\v r\'ak's theorems follow from it.

Before stating this general result, we illustrate the underlying idea for a generic game comonad $\CC$.
In Section~\ref{s:prelim} we have recalled how game comonads capture combinatorial parameters of $\sigma$-structures. On the other hand, they can also express equivalence with respect to appropriate logic fragments. For this, we need to extend the signature to account for equality in the logic. Let $\sigma^+\coloneqq \sigma\cup \{I\}$ be the relational signature obtained by adding a binary relation symbol $I$ to $\sigma$, and let $\Si^+$ be the category of $\sigma^+$-structures with homomorphisms. There is an \emph{embedding}, i.e.\ a full and faithful functor, 
\[ J\colon \Si\to \Si^+ \]
which sends a $\sigma$-structure $A$ to the $\sigma^+$-structure $J(A)$ where $I$ is interpreted as the identity.

As the comonad $\CC$ on $\Si$ is defined for arbitrary relational signatures, we have a corresponding comonad $\CC^+$ on $\Si^+$. Typically, it turns out that, for an appropriate logic fragment~$\mathcal L$,
\[
A\equiv_{\mathcal L} B \ \ \Longleftrightarrow \ \  F^{\CC^+}(J(A)) \cong F^{\CC^+}(J(B))
\]
for all $A,B\in \Si_f$, where $U^{\CC^+} \dashv F^{\CC^+}\colon\Si^+ \to \EM(\CC^+)$ is the adjunction associated with $\CC^+$, cf.\ Section~\ref{s:prelim-coalgebras}.

Provided that $\CC^+$ sends finite structures to finite structures, we can apply Corollary~\ref{cor:combinatorial-EM-f} to translate the isomorphism $F^{\CC^+}(J(A)) \cong F^{\CC^+}(J(B))$ into a statement about homomorphism counts for $J(A)$ and $J(B)$ in $\Si_f^+$. It then remains to go back to $\sigma$-structures to obtain a statement about homomorphism counts for $A$ and $B$ in $\Si_f$. To this end, note that the functor $J\colon \Si\to \Si^+$ has a left adjoint 
\[
H\colon \Si^+ \to \Si
\] 
sending a $\sigma^+$-structure $D$ to $D^-/{\sim}$, where $D^-$ is the $\sigma$-reduct of $D$ and $\sim$ is the equivalence relation generated by $I^D$ (for a proof of this fact, see Lemma~\ref{l:adjunction-HJ}).

We now state the general result (whose proof is deferred to Section~\ref{s:proof-of-general-statement}) from which Grohe's and Dvo\v r\'ak's theorems will be derived in Sections~\ref{s:hom-count-Grohe} and~\ref{s:tw}, respectively. Let us write $\im(U^{\CC})$ for the full subcategory of $\Si$ consisting of the objects of the form $U^\CC(A,\alpha) = A$ for $(A,\alpha)\in \EM(\CC)$, and similarly for $\im(U^{\CC^+})$.

\begin{theorem}\label{t:general-statement}
Let $\CC$ and $\CC^+$ be comonads on $\Si$ and $\Si^+$, respectively, and $\mathcal L$ a logic fragment. Suppose the following conditions are satisfied:
\begin{enumerate}
    \item\label{logic-frag} for all $A,B\in\Si_f$, $A \equiv_{\mathcal L} B$ if, and only if, $F^{\CC^+}(J(A)) \cong F^{\CC^+}(J(B))$;
     \item\label{restr-fin} $\CC^+$ sends finite $\sigma^+$-structures to finite $\sigma^+$-structures;
    \item\label{combinat} the embedding $J\colon \Si \to \Si^+$ and its left adjoint $H$ restrict to $\Si_f \cap \im(U^\CC)$ and $\Si_f^+ \cap \im(U^{\CC^+})$.
\end{enumerate}
Then, for any finite $\sigma$-structures $A$ and $B$,
\[  A \equiv_{\mathcal L} B \ \text{ if, and only if, } \ |\hom_{\Si_f}(C,A)| = |\hom_{\Si_f}(C,B)| \]
for every finite $\sigma$-structure $C$ in $\im(U^\CC)$.
\end{theorem}

\begin{remark}
The proofs of Grohe's and Dvo\v r\'ak's theorems essentially reduce to showing that the assumptions of Theorem~\ref{t:general-statement} are satisfied for the appropriate comonads. The ``combinatorial core'' of these results, which requires a specific argument for each comonad, corresponds to verifying that the functor $H$ restricts to $\Si_f^+ \cap \im(U^{\CC^+})\to \Si_f \cap \im(U^\CC)$. This amounts to saying that the operation $D\mapsto D^-/{\sim}$ does not increase the tree-depth or tree-width of $D$, and can be understood as an equality elimination result (cf.\ Section~\ref{s:normal-forms}).
\end{remark}

\subsection{Bounded tree-depth}\label{s:hom-count-Grohe}
Let $\CL$ be the extension of first-order logic obtained by adding, for each natural number $i$, a counting quantifier~$\exists_{{\geq}i}$. The semantics of these quantifiers is the following: for any structure $A$, we have $A \models \exists_{{\geq}i} x.\, \varphi$ if, and only if, $A \models \varphi[a/x]$ holds for at least $i$ distinct elements $a\in A$. Let $\CL_n$ be the fragment of $\CL$ consisting of formulas of quantifier depth ${\leq \,}n$. 

The aim of this section is to use Theorem~\ref{t:general-statement} to give a new proof of a recent result of Grohe~\cite{grohe2020counting}:\footnote{Grohe proved this result for undirected vertex-coloured graphs, but pointed out that his proof can be extended to arbitrary $\sigma$-structures.}
\begin{theorem}\label{thm:Grohe}
    For any $A,B$ finite $\sigma$-structures,
    \[ A\equiv_{\CL_n} B \qtq{if, and only if,} |\hom_{\Si_f}(C,A)|= |\hom_{\Si_f}(C,B)|\]
    for every finite $\sigma$-structure $C$ with tree-depth at most $n$.
\end{theorem}

Let $\En^+$ be the {\ef} comonad on $\Si^+$, and $U^{\En^+}\dashv F^{\En^+}\colon \Si^+\to \EM(\En^+)$ the associated adjunction (cf.\ Section~\ref{s:prelim-coalgebras}). 
By Proposition~\ref{p:En-coalgebras}, a $\sigma$-structure has tree-depth at most $n$ precisely when it belongs to $\im(U^{\En})$. Thus, to prove Theorem~\ref{thm:Grohe}, it suffices to show that the assumptions of Theorem~\ref{t:general-statement} are satisfied for $\CC=\En$, $\CC^+=\En^+$, and $\mathcal{L}=\CL_n$.

Item~\ref{logic-frag} in Theorem~\ref{t:general-statement} is a consequence of~\cite[Theorem~12]{AbramskyShah2018}:
\begin{proposition}\label{pr:logic-EM-iso}
For any two $\sigma$-structures $A,B$, we have  $A\equiv_{\CL_n} B$ if, and only if, $F^{\En^+}(J(A)) \cong F^{\En^+}(J(B))$.
\end{proposition}
Item~\ref{restr-fin} holds because the comonad $\En^+$ restricts to finite $\sigma^+$-structures, and item~\ref{combinat} is a consequence of the following proposition---which concludes the proof of Grohe's theorem.
\begin{proposition}\label{pr:adjunction-HJ-td-k}
    The embedding $J\colon \Si \to \Si^+$ and its left adjoint $H$ restrict to the full subcategories consisting of finite structures with tree-depth at most $n$.
\end{proposition}
\begin{proof}
Clearly, if $A$ is a (finite) $\sigma$-structure with tree-depth at most $n$, then $J(A)$ has tree-depth at most $n$. It remains to prove that, for any $D\in\Si_f^+$, if $D$ has tree-depth at most $n$ then so does $H(D)$.
To improve readability, set $C\coloneqq H(D)$. 

Suppose $(F,\leq)$ is a forest of height ${\leq \,} n$ and $f\colon D\to F$ an injective map such that $f(d_1)\up f(d_2)$ whenever $d_1\adj d_2$ in the Gaifman graph $\G_D$. Recall from the definition of $H$ that, at the level of sets, $C=D/{\sim}$ where $\sim$ is the equivalence relation generated by $I^D$. If $d\in D$, let $[d]$ denote its $\sim$-equivalence class. For each $d\in D$, set
\[
\xi(d)\coloneqq\min{\{f(d')\mid d'\in [d]\}}.
\]
This minimum exists because $\{f(d')\mid d'\in [d]\}$ is a connected subset of $F$. (Alternatively, note that $\xi(d)$ coincides with the minimum of $\{f(d')\mid d'\in [d]\}\cap \down  f(d)$, which is a finite non-empty totally ordered set.)
Define 
\[
g\colon C\to F, \ \ g([d])\coloneqq \xi(d). 
\]
Clearly, $g$ is well-defined and it is injective because so is $f$.
It remains to prove that, for all $d_1,d_2\in D$, $[d_1]\adj [d_2]$ in $\G_C$ entails $g([d_1])\up g([d_2])$, for then it follows that $(F,\leq)$, along with the map $g$, is a forest cover of $\G_C$.

If $[d_1]\adj [d_2]$, there are $d'_1\in [d_1]$ and $d'_2\in [d_2]$ satisfying $d'_1\adj d'_2$, and so $f(d'_1)\up f(d'_2)$. Assume $f(d'_1)\leq f(d'_2)$. Since $g([d_1])\leq f(d'_1)$ and $g([d_2])\leq f(d'_2)$, it follows that both $g([d_1])$ and $g([d_2])$ belong to $\down f(d'_2)$, which is totally ordered. Therefore, $g([d_1])\up g([d_2])$. Reasoning in a similar manner, we see that $f(d'_2)\leq f(d'_1)$ implies $g([d_1])\up g([d_2])$.
\end{proof}

\subsection{Bounded tree-width}\label{s:tw}
Let $\CL^k$ be the $k$-variable fragment of $\CL$---first-order logic with counting quantifiers.
In this section, we show how to derive from Theorem~\ref{t:general-statement} the following variant of Dvo\v r\'ak's theorem~\cite{dvovrak2010recognizing}:\footnote{Dvo\v r\'ak's result is for undirected graphs without loops. For a discussion of how his result can be recovered in our framework, see Remark~\ref{rm:counting-tree-width-subclasses}.}
\begin{theorem}\label{thm:Dvorak}
    For any $A,B$ finite $\sigma$-structures,
    \[ A\equiv_{\CL^k} B \qtq{if, and only if,} |\hom_{\Si_f}(C,A)|= |\hom_{\Si_f}(C,B)|\]
    for every finite $\sigma$-structure $C$ with tree-width at most $k-1$.
\end{theorem}

The $k$-variable counting logic $\CL^k$ appearing in Theorem~\ref{thm:Dvorak} is captured by the pebbling comonad $\Pk^+$ on $\Si^+$. In fact, it follows directly from~\cite[Theorem~18]{Abramsky2017b} that, whenever $A,B$ are finite $\sigma$-structures, $F^{\Pk^+}(J(A)) \cong F^{\Pk^+}(J(B))$ if, and only if, $A\equiv_{\CL^k} B$.
However, the comonad $\Pk^+$ does not satisfy item~\ref{restr-fin} in Theorem~\ref{t:general-statement}, because $\Pk^+(A)$ is infinite even when $A\in\Si^+_f$.

To circumvent this problem, we consider the comonads $\Pkn^+$, which do restrict to finite $\sigma^+$-structures. Again by (the proof of) \cite[Theorem~18]{Abramsky2017b}, for any finite $\sigma$-structures $A,B$,
\[
A\equiv_{\CL^k_n} B \ \ \Longleftrightarrow \ \  F^{\Pkn^+}(J(A)) \cong F^{\Pkn^+}(J(B))
\]
where $\CL_n^k$ consists of the formulas that are simultaneously in $\CL_n$ and $\CL^k$. Thus, items~\ref{logic-frag} and~\ref{restr-fin} in Theorem~\ref{t:general-statement} are satisfied for $\CC=\Pkn$, $\CC^+=\Pkn^+$, and $\mathcal{L}=\CL_n^k$.

Assume for a moment that item~\ref{combinat} is also satisfied. Then, combining Theorem~\ref{t:general-statement} with Proposition~\ref{p:Pkn-coalgebras}, we obtain that 
\[
A\equiv_{\CL^k_n} B \ \text{ if, and only if, } \ |\hom_{\Si_f}(C,A)| = |\hom_{\Si_f}(C,B)|
\]
for all finite $\sigma$-structures $C$ whose Gaifman graph $\G_C$ admits a $k$-pebble forest cover of height ${\leq \,} n$. Observe next that (i) ${A\equiv_{\CL^k} B}$ precisely when $A\equiv_{\CL^k_n} B$ for all $n$, and (ii) a finite $\sigma$-structure $C$ has tree-width at most $k-1$ if, and only if, $\G_C$ admits a $k$-pebble forest cover of height ${\leq \,} n$ for some $n$. Therefore, $A\equiv_{\CL^k} B$ if, and only if, 
\[
|\hom_{\Si_f}(C,A)|= |\hom_{\Si_f}(C,B)| 
\]
for all finite $\sigma$-structures $C$ with tree-width at most $k-1$, thus settling Theorem~\ref{thm:Dvorak}.

In the remainder of this section, we prove that item~\ref{combinat} in Theorem~\ref{t:general-statement} is indeed satisfied. As mentioned already, $\Si_f\cap \im(U^{\Pkn})$ consists of the finite $\sigma$-structures that admit a $k$-pebble forest cover of height ${\leq \,}n$, and similarly for $\Si^+_f\cap \im(U^{\Pkn^+})$. To improve readability, let us denote these categories by $\Tkn$ and $\Tkn^+$, respectively.
Note that $J\colon \Si_f \to \Si_f^+$ restricts to a functor $\Tkn \to \Tkn^+$. It remains to show that its left adjoint $H$ restricts to $\Tkn^+ \to \Tkn$.

Recall from Section~\ref{s:prelim-coalgebras} that a forest cover of a $\sigma$-structure $A$ can be identified with a forest order $\leq$ on $A$ that is \emph{compatible}, in the sense that $a\up a'$ whenever $a\adj a'$ in $\G_A$. Suppose that $A$ is equipped with a forest order $\leq$ and let $p\colon A\to \k$ be any function. We define a relation $\sees$ on $A$ as follows: $a\sees a'$ (read as \emph{$a$ sees $a'$}) if $a\up a'$ and 
\[
\min(a,a')<z\leq \max(a,a') \ \Longrightarrow \ p(z)\neq p(\min(a,a'))
\]
for all $z\in A$. With this notation, a $k$-pebble forest cover of $A$ can be identified with a triple $(A,\leq,p)$ where $\leq$ is a forest order on $A$ and $p\colon A\to \k$ is such that $a\adj a'$ implies $a\sees a'$ for all $a,a'\in A$. Just observe that $\leq$ is a compatible forest order since $a\sees a'$ entails $a\up a'$.

Saying that $H$ restricts to $\Tkn^+ \to \Tkn$ means that, whenever a finite $\sigma^+$-structure $A$ admits a $k$-pebble forest cover of height ${\leq \,}n$, then so does the quotient $A^-/{\sim}$ where $A^-$ is the $\sigma$-reduct of $A$ and $\sim$ is the least equivalence relation containing $I^A$. The difficulty in showing this consists in defining the pebbling function on $A^-/{\sim}$. We do this by considering consecutive one-step quotients (inspired by equality elimination from Section~\ref{s:normal-forms}), where at each step a new pair of distinct elements in $I^A$ is identified. As $A$ is finite, this construction terminates after finitely many steps and yields the quotient $A^-/{\sim}$. 

Let $\mathbb A\coloneqq (A,\leq,p)$ be a finite $\sigma^+$-structure together with a $k$-pebble forest cover of height ${\leq \,}n$. Suppose that there is a pair $(u,v) \in I^A$ such that $u \neq v$. Then $u\up v$, and we can assume without loss of generality that $u < v$. The one-step quotient $\mathbb A/_{u \sim v} \coloneqq (A', \leq', p')$ of $\mathbb A$ is defined as follows:

\begin{itemize}
    \item the carrier of $\mathbb A/_{u \sim v}$ is the set $A' \coloneqq A \setminus \{v\}$;
    \item the forest order $\leq'$ is the restriction of $\leq$ to $A'$;
    \item for each $R\in\sigma\cup\{I\}$, the relation $R^{A'}$ is obtained by replacing all occurrences of $v$ by $u$ in each tuple of $R^A$. 
    \item the pebbling function $p'\colon A' \to \k$ is defined by
        \[ w \mapsto
    \begin{cases}
        p(u) & \text{if } v \leq w,\ p(w) = p(v),\ \text{and } u \not\sees w_{\min} \\
        p(v) & \text{if } v \leq w,\ p(w) = p(u),\ \text{and } v \sees w_{\min} \\
        p(w) & \text{otherwise}
    \end{cases}
    \]
    where $w_{\min} \coloneqq \min \{w' \in A \mid v < w' \leq w \text{ and } p(w') = p(w)\}$.
\end{itemize}
Clearly, $\leq'$ is a forest order of height ${\leq \,}n$ on $A'$. Further, $a\adj a'$ implies $a\sees a'$ for all $a,a'\in A'$. This follows from the next lemma, which is proved by a careful case analysis. 
\begin{lemma}\label{l:about-sees-relation}
    Let $w, w' \in A \setminus \{v\}$. Then the following hold:
    \begin{enumerate}
        \item $v \sees w$ in $\mathbb A$ implies $u \sees w$ in $\mathbb A/_{u \sim v}$;
        \item $w \sees w'$ in $\mathbb A$ implies $w \sees w'$ in $\mathbb A/_{u \sim v}$.
    \end{enumerate}
\end{lemma}
\begin{proof}
    (1) Suppose $w\in A\setminus \{v\}$ and $v\sees w$ in $\mathbb A$. There are two possibilities, either $v < w$ or $w < v$. In both cases we have that $w \up u$ and so we only need to check the condition on the pebbling function $p'$. Assume that $v < w$. As $(u,v) \in I^A$, we know that $u\adj v$ and so $u \sees v$ in $\mathbb A$. Because $p'(z)=p(z)$ whenever $v\not\leq z$, to show that $u \sees w$ in $\mathbb A/_{u \sim v}$ we only need to check that
    \begin{equation}\label{eq:pebbling-cond-A'}
    \forall z\in A, \ \ v < z \leq w \ \Longrightarrow \ p'(z) \not= p'(u).
    \end{equation}
    Observe that $p'(u) = p(u)$. Also, by definition of $p'$, for any $z$ such that $v < z \leq w$ the equality $p'(z) = p(u)$ is satisfied if, and only if, either $p(z) = p(v)$ and $u \not\sees z_{\min}$, or $p(z) = p(u)$ and $v \not\sees z_{\min}$. However, none of these two scenarios can occur because $v \sees w$ implies $p(z) \neq p(v)$ and $v \sees z_{\min}$. Hence, the condition in~\eqref{eq:pebbling-cond-A'} is satisfied.

    On the other hand, assume that $w < v$. Then $u\up w$. If $u < w$ then $u \sees v$ implies $u \sees w$ in $\mathbb A$. Similarly, if $w \leq u$ then $v \sees w$ entails $u \sees w$ in $\mathbb A$. Therefore, $u \sees w$ in $\mathbb A/_{u \sim v}$ because $p(z) = p'(z)$ for every $z$ such that $v \not\leq z$.

    (2) Let $w,w'\in A\setminus \{v\}$ be such that $w\sees w'$ in $\mathbb A$. We can assume without loss of generality that $w \leq w'$. If $v \not\leq w'$ then $p(z) = p'(z)$ for every $z$ such that $w \leq z \leq w'$, and so $w\sees w'$ in $\mathbb A/_{u \sim v}$. 
    Thus, suppose $v \leq w'$. It follows that $v\up w$ and so either $v < w$ or $w < v$. 
    
    First we assume $v<w \leq w'$. Note that, for any $z$ such that $w<z\leq w'$, if $p'(z)=p'(w)$ and $p(z)\notin \{p(u),p(v)\}$, then by definition of $p'$ we must have $p(z)=p(w)$, contradicting the fact that $w\sees w'$ in $\mathbb{A}$. So, we only have to take care of the case where either $p(z)=p(u)$ or $p(z)=p(v)$. Set
    \begin{gather*}
    U\coloneqq \{x\in A\mid v < x \leq w' \text{ and } p(x)=p(u)\} \\
    V\coloneqq \{x\in A\mid v < x \leq w' \text{ and } p(x)=p(v)\}.
    \end{gather*}
    Suppose $U$ and $V$ are non-empty and so they admit minimal elements $\o{u}$ and $\o{v}$, respectively. Then, either $\o{u}<\o{v}$ or $\o{v}<\o{u}$. If $\o{u}<\o{v}$ then any $z\geq \o{u}$ such that $p(z)=p(u)$ satisfies $v\sees z_{\min}$ (because $z_{\min}=\o{u}$) and so $p'(z)=p(v)$, and similarly any $z\geq \o{u}$ such that $p(z)=p(v)$ satisfies $u\not\sees z_{\min}$ (because $z_{\min}=\o{v}$) and so $p'(z)=p(u)$.
  On the other hand, if $\o{v}<\o{u}$ then $p'(z)=p(z)$ for any $z\geq \o{v}$. Either way, for any $z$ between $w$ and $w'$, $p'(z)=p'(w)$ implies $p(z)=p(w)$. Hence, $w\sees w'$ in $\mathbb A$ entails $w \sees w'$ in $\mathbb A/_{u \sim v}$. A similar reasoning applies, mutatis mutandis, when $U=\emptyset$ or $V=\emptyset$.

Suppose next that $w < v < w'$. Then $p'(w)=p(w)$ and, as $w \sees w'$ in $\mathbb A$, it follows that $p(w) \not= p(v)$. If $p(w) \not= p(u)$ then we conclude that $w \sees w'$ in $\mathbb A/_{u \sim v}$ because $p'(z)=p'(w)$ entails $p(z)=p(w)$ for any element $z$ between $w$ and $w'$. 
To conclude the proof, suppose towards a contradiction that $p(w) = p(u)$ and there exists a $z\in A\setminus\{v\}$ such that $w < z \leq w'$ and $p'(z) = p'(w)$. Then, $p'(z)=p(u)$. Since $w \sees w'$ in $\mathbb A$, we get $p(z) \not= p(u)$ and so $v< z$, $p(z)=p(v)$ and $u \not\sees z_{\min}$ in $\mathbb A$.
As $u \sees v$ in $\mathbb A$, there must exist an $x$ such that
    \[ w < v < x \leq z_{\min} \leq z \leq w'\]
    and $p(x) = p(u)$, contradicting the fact that $w \sees w'$ in $\mathbb A$.
\end{proof}

Applying this construction iteratively, we obtain the following analogue of Proposition~\ref{pr:adjunction-HJ-td-k}, which concludes the proof of Dvo\v r\'ak's theorem.
\begin{proposition}\label{p:td-and-tw-preserved}
   The functor $H$ restricts to $\Tkn^+ \to \Tkn$.
\end{proposition}
\begin{proof}
Let $\mathbb A_0=(A,\leq,p)$ be a finite $\sigma^+$-structure together with a $k$-pebble forest cover of height ${\leq \,}n$. Then we have a sequence $\mathbb A_0,\mathbb A_1,\mathbb A_2,\dots$ such that $\mathbb A_{i+1}$ is defined as $\mathbb A_i/_{u\sim v}$ for some pair of distinct elements $(u,v)\in I^{A_i}$, where $A_i$ is the $\sigma^+$-structure underlying $\mathbb A_i$. This is iterated as long as such a pair $(u,v)$ can be found. Let $\mathbb A_n$ be the last element of the sequence. Note that, by construction, the $\sigma$-reduct of $A_n$ is isomorphic to $H(A)$. Because $A_n$ admits a $k$-pebble forest cover of height ${\leq \,}n$ witnessed by $\mathbb A_n$---by repeatedly applying Lemma~\ref{l:about-sees-relation}---so does $H(A)$.
\end{proof}

\begin{remark}\label{rm:counting-tree-width-subclasses}
An advantage of the categorical approach to homomorphism counting is that it specialises to any full subcategory $\A$ of $\Si$, provided the game comonad in question restricts to $\A$. For example, let $\sigma$ consist of a single binary relation symbol, so that $\Si$ is the category of (directed) graphs and graph homomorphisms. As the pebbling comonad $\Pk$ restricts to the full subcategory of $\Si$ defined by undirected graphs without loops, we obtain an analogue of Theorem~\ref{thm:Dvorak} where $A,B,C$ range over the class of finite undirected graphs without loops. A similar observation applies to the {\ef} comonad $\En$, yielding a variant of Theorem~\ref{thm:Grohe} for undirected graphs without loops.
\end{remark}

\subsection{A proof of Theorem~\ref{t:general-statement}}\label{s:proof-of-general-statement}

We start by showing that the functor $H\colon \Si^+\to \Si$ sending $D$ to $D^-/{\sim}$, where $D^-$ is the $\sigma$-reduct of $D$ and $\sim$ is the equivalence relation generated by $I^D$, is left adjoint to $J$.

\begin{lemma}\label{l:adjunction-HJ}
$H\colon \Si^+\to \Si$ is left adjoint to the embedding $J$.
\end{lemma}
\begin{proof}
It suffices to show that for every $D\in\Si^+$, the unique $\sigma^+$-homomorphism $e_D \colon D\to JH(D)$ extending the quotient map $D^-\twoheadrightarrow H(D)$ satisfies the following universal property: For every $\sigma^+$-homomorphism $j\colon D\to J(C)$, there exists a unique $\sigma$-homomorphism $h\colon H(D)\to C$ making the following diagram commute.
\[\begin{tikzcd}
D \arrow{r}{e_D} \arrow{dr}[swap]{j} &JH(D) \arrow{d}{J(h)} \\
{} & J(C)
\end{tikzcd}\]

Pick an arbitrary $\sigma^+$-homomorphism $j\colon D\to J(C)$. Since $j$ preserves the relation $I$, $d_1{\sim} d_2$ entails $j(d_1)=j(d_2)$ for all $d_1,d_2\in D$. Hence, the following map is well defined 
\[
h\colon H(D)\to C, \ \ h([d])\coloneqq j(d)
\] 
where $[d]$ denotes the $\sim$-equivalence class of $d$. Moreover, $h$ is a $\sigma$-homomorphism. Just observe that, for any relation symbol $R\in \sigma$ of arity $n$, if $([d_1],\ldots,[d_n])\in R^{H(D)}$ then there are $d'_1\sim d_1,\ldots, d'_n\sim d_n$ such that $(d'_1,\ldots, d'_n)\in R^{D}$. Thus, $(h([d_1]),\ldots,h([d_n]))=(j(d'_1),\ldots, j(d'_n))\in R^{C}$. Further, $j=J(h)\circ e_D$ by construction. To conclude, note that $h$ is unique with this property because $e_D$ is surjective.
\end{proof}
\begin{remark}
Since both $J$ and $H$ send finite structures to finite structures, they restrict to an adjunction $H\dashv J\colon \Si_f\to\Si_f^+$.
\end{remark}

Before proceeding to the proof of Theorem~\ref{t:general-statement}, we make the following easy observation.
\begin{lemma}\label{l:I-quot}
Assume that $J$ and $H$ restrict to full subcategories $\A \subseteq \Si_f$ and $\A^+ \subseteq \Si_f^+$. Then, for any $A,B\in\Si_f$, the following are equivalent:
\begin{enumerate}
\item\label{hom-bij-sigma} $\hom_{\Si_f}(C,A)\cong \hom_{\Si_f}(C,B) \ \ \forall C\in \A$;
\item\label{hom-bij-sigmap} $\hom_{\Si_f^+}(C',J(A))\cong\hom_{\Si_f^+}(C',J(B)) \ \ \forall C'\in\A^+$.
\end{enumerate}
\end{lemma}
\begin{proof}
Assume item~\ref{hom-bij-sigma} holds. For any $C'\in \A^+$,
\begin{align*}
\hom_{\Si_f^+}(C',J(A)) &\cong \hom_{\Si_f}(H(C'),A) \tag{$H\dashv J$}  \\
& \cong \hom_{\Si_f}(H(C'),B) \tag{$H(C')\in \A$} \\
&\cong \hom_{\Si_f^+}(C',J(B)). \tag{$H\dashv J$}
\end{align*}

 Conversely, suppose item~\ref{hom-bij-sigmap} holds. Then, for any $C\in\A$, 
\begin{align*}
\hom_{\Si_f}(C,A) &\cong \hom_{\Si_f^+}(J(C),J(A)) \tag{$J$ embedding}  \\
& \cong \hom_{\Si_f^+}(J(C),J(B)) \tag{$J(C)\in \A^+$} \\
&\cong \hom_{\Si_f}(C,B). \tag{$J$ embedding}
\end{align*}
This concludes the proof.
\end{proof}

\begin{proof}[Proof of Theorem~\ref{t:general-statement}]
For convenience of notation, let $\EM_f^+\coloneqq \EM_f(\CC^+)$.
Combining item~\ref{logic-frag} with Corollary~\ref{cor:combinatorial-EM-f} applied to the comonad $\CC^+$, we see that, for all $A,B\in\Si_f$, $A\equiv_{\mathcal L} B$ if, and only if,
\begin{equation}\label{eq:hom-bij-co-free-gen}
|\hom_{\EM_f^+}(D,F^{\CC^+}(J(A)))|=|\hom_{\EM_f^+}(D,F^{\CC^+}(J(B)))| 
\end{equation}
for all finite coalgebras $D\in\EM_f^+$. Note that here we used the fact that, by item~\ref{restr-fin}, $F^{\CC^+}(D')$ is a finite coalgebra whenever $D'$ is a finite $\sigma^+$-structure. In view of the adjunction $U^{\CC^+}\dashv F^{\CC^+}$, the condition in equation~\eqref{eq:hom-bij-co-free-gen} is equivalent to
\[
|\hom_{\Si^+_f}(U^{\CC^+}(D),J(A))|=|\hom_{\Si^+_f}(U^{\CC^+}(D),J(B))| 
\]
for all $D\in\EM_f^+$. Finally, by item~\ref{combinat}, an application of Lemma~\ref{l:I-quot} with $\A\coloneqq \Si_f \cap \im(U^\CC)$ and $\A^+\coloneqq \Si_f^+ \cap \im(U^{\CC^+})$ yields the desired statement.
\end{proof}

\section{Homomorphism Counting in Modal Logic}\label{s:applications-modal}

In this section we prove a new homomorphism counting result for (multi-)modal logic. This is derived by applying the categorical framework from Section~\ref{s:categorical-Lovasz} to a game comonad for modal logic defined on the category of pointed Kripke structures. Since there is no equality in the logic, this time we can dispense with the extra relation $I$. This considerably simplifies the proofs, compared to the previous section.

Let $\sigma$ be a signature consisting of relation symbols of arity at most $2$. Each unary relation symbol $P$ yields a propositional variable $p$, and each binary relation symbol $R_{\alpha}$ yields modalities $\Box_{\alpha}$ and $\Diamond_{\alpha}$.
We can think of a $\sigma$-structure $A$ as a Kripke structure for this multi-modal logic, where $P^A$ gives the valuation for the propositional variable $p$, and $R_{\alpha}^A$ gives the accessibility relation for the modalities $\Box_{\alpha}$ and $\Diamond_{\alpha}$.

Each formula $\phi$ in this modal logic admits a translation into a first-order formula $\tr{\phi}_x$ in one free variable $x$; this is known as the \emph{standard translation}, see e.g.~\cite[\S 2.4]{BdRV2001}. We let $\tr{p}_x\coloneqq P(x)$ and let $\tr{\ }_x$ commute with Boolean connectives. Further, set $\tr{\Box_{\alpha}\phi}_x\coloneqq \forall y. \, R_{\alpha}(x,y) \to \tr{\phi}_y$ and $\tr{\Diamond_{\alpha}\phi}_x\coloneqq \exists y. \, R_{\alpha}(x,y) \wedge \tr{\phi}_y$ where $y$ is a fresh variable. Consider a \emph{pointed Kripke structure}, that is a pair $(A,a)$ where $A$ is a $\sigma$-structure and $a\in A$. Then $A,a\models \phi$ according to Kripke semantics if, and only if, $A\models \tr{\phi}_x [a/x]$ in the standard model-theoretic sense.

Henceforth, we assume all Kripke structures under consideration are pointed. Let $\kri$ be the category of Kripke structures and $\sigma$-homomorphisms preserving the distinguished element.
For every natural number $k$, there is a \emph{modal comonad} $\Mk$ on $\kri$, see~\cite{AbramskyShah2018}. For any Kripke structure $(A,a)$, the carrier of $\Mk(A,a)$ is the set of all paths of length ${\leq \,}k$ starting from $a$:
\[ a \xrightarrow{R_1} a_1 \xrightarrow{R_2} a_2 \to \cdots \xrightarrow{R_n} a_n \]
where $R_1, \dots, R_n$ are binary relations from $\sigma$.  
The distinguished element of the Kripke structure $\Mk(A,a)$ is the trivial path $(a)$ of length~$0$. If $P\in\sigma$ is unary then $P^{\Mk(A,a)}$ is defined as the set of paths $p$ such that the last element $a_n$ of $p$ belongs to $P^{A}$. For a binary relation symbol $R\in\sigma$, let $R^{\Mk(A,a)}$ be the set of pairs of paths $(p,p')$ such that $p'$ is obtained by extending $p$ by one step along $R$. The morphism $\epsilon_{(A,a)}\colon \Mk(A,a) \to (A,a)$ sends a path to its last element; for the definition of the coextension operation, see~\cite[\S 2.3]{AbramskyShah2018}.

The comonad $\Mk$ captures a well known combinatorial parameter of Kripke structures, as we now recall. Let us say that a Kripke structure $(A,a)$ is \emph{rooted} if, for any $a'\in A$, there is a path from $a$ to $a'$. If we further require that this path be unique, then $(A,a)$ is a \emph{synchronization tree}. A synchronization tree $(A,a)$ has \emph{height ${\leq \,}k$} if, for each $a'\in A$, the length of the unique path from $a$ to $a'$ is at most $k$. For a proof of the following result, see \cite[Proposition 6.6]{AbramskyShah-ext}.

\begin{proposition}\label{p:Mk-structure}
    A rooted Kripke structure $(A,a)$ admits a coalgebra structure $(A,a) \to \Mk(A,a)$ if, and only if, $(A,a)$ is a synchronization tree of height ${\leq \,}k$.
\end{proposition}

On the other hand, the modal comonad $\Mk$ characterises logical equivalence with respect to an appropriate modal logic. Let us extend the multi-modal logic introduced above by adding, for every positive integer $n$ and binary relation symbol $R_{\alpha}\in\sigma$, \emph{graded modalities} $\Box_{\alpha}^n$ and $\Diamond_{\alpha}^n$ defining $A,a\models \Diamond_{\alpha}^n\phi$ if there are at least $n$ distinct $R_{\alpha}$-successors of $a$ that satisfy~$\phi$ (and $\Box_{\alpha}^n\phi=\neg \Diamond_{\alpha}^n\neg\phi$). A formula in this graded modal logic has \emph{modal depth ${\leq \,}k$} if it contains at most $k$ nested modalities. Let $\CMLk$ be the image of the standard translation of formulas with graded modalities that have modal depth ${\leq \,}k$. 
The following result is a special case of~\cite[Theorem~16]{AbramskyShah2018}. 
\begin{proposition}\label{p:Mk-logic}
    For finite Kripke structures $(A,a)$ and $(B,b)$,
    \[ (A,a) \equiv_{\CMLk} (B,b) \quad\text{iff}\quad F^{\Mk}(A,a) \,\cong\, F^{\Mk}(B,b). \]
\end{proposition}

We aim to exploit the two preceding propositions to obtain a characterisation of the equivalence relation $\equiv_{\CMLk}$ over finite Kripke structures in terms of homomorphism counts from finite synchronization trees of height ${\leq\,}k$ (Theorem~\ref{thm:modal-counting} below).

To this end, let $\kri_f$ and $\rkri_f$ be the full subcategories of $\kri$ consisting, respectively, of the finite Kripke structures, and finite rooted Kripke structures. Note that $\Mk$ restricts to a comonad on $\kri_f$, that we denote again by $\Mk$.
Further, as $\Mk(A,a)$ is a finite synchronization tree whenever $(A,a)$ is a finite Kripke structure, $\Mk$ restricts to a comonad $\Mk^*$ on the category $\rkri_f$. 
\begin{lemma}\label{l:rooting}
The following statements hold:
    \begin{enumerate}
        \item the inclusion $\rkri_f \hookrightarrow \kri_f$ has a right adjoint $R$;
        \item $R$ lifts to a functor $\overline R\colon \EM(\Mk)\to \EM(\Mk^*)$  making the following diagram commute.
        \[
            \begin{tikzcd}
                \EM(\Mk)\rar{\overline R} & \EM(\Mk^*) \\
                \kri_f\ar{u}{F^{\Mk}}\rar{R} & \rkri_f\ar[swap]{u}{F^{\Mk^*}}
            \end{tikzcd}
        \]
    \end{enumerate}
\end{lemma}
\begin{proof}
    (1) For any finite Kripke structure $(A,a)$, let $R(A,a)$ be the substructure of $(A,a)$ consisting of the elements accessible from $a$. To show that this assignment extends to a functor $R$ that is a right adjoint to the inclusion $\rkri_f \hookrightarrow \kri_f$, it is enough to observe that, for any $(A,a)\in \kri_f$, $(B,b)\in \rkri_f$, and homomorphism $h\colon (B,b) \to (A,a)$, there exists a unique homomorphism $h'\colon (B,b)\to R(A,a)$ such that $h = e \circ h'$, where $e\colon R(A,a) \to (A,a)$ is the substructure embedding.

    (2) Let $\overline R$ be the functor sending a coalgebra $\alpha\colon (A,a) \to \Mk (A,a)$ to its restriction $\alpha_{\restriction R(A,a)} \colon R(A,a)\to \Mk^*R(A,a)$. Then $\overline R F^{\Mk}(A,a)$ coincides with $F^{\Mk^*}R(A,a)$, for any $(A,a)$ in $\kri_f$, because the carrier of $F^{\Mk}(A,a) = (\Mk(A,a), \delta_A)$ is rooted and consists only of paths of length ${\leq\,}k$ that start at~$a$. The equality on morphisms holds for the same reason.
\end{proof}

We are now ready to prove a Lov\'asz-type theorem for graded modal logic.
\begin{theorem}\label{thm:modal-counting}
    Let $(A,a)$ and $(B,b)$ be finite (pointed) Kripke structures. Then $(A,a) \equiv_{\CMLk} (B,b)$ if, and only if,
    \[ |\hom_{\kri_f}((C,c),(A,a))| \,=\, |\hom_{\kri_f}((C,c),(B,b))| \]
    for every finite synchronization tree $(C,c)$ of height ${\leq \,} k$.
\end{theorem}
\begin{proof}
    Let $\mathbb A \coloneqq (A,a)$ and $\mathbb B \coloneqq (B,b)$. For the purpose of this proof, write $V\colon \rkri_f \into \kri_f$ for the inclusion functor and let $\EM^*\coloneqq\EM(\Mk^*)$. By Proposition~\ref{p:Mk-structure}, the statement about homomorphism counts is equivalent to
    \[ |\hom_{\kri_f}(V(U^{\Mk^*}(\mathbb D)),\mathbb A)| \,=\, |\hom_{\kri_f}(V(U^{\Mk^*}(\mathbb D)), \mathbb B)| \]
    for all coalgebras $\mathbb D$ in $\EM^*$. In turn, by virtue of the adjunctions $V \dashv R$ and $U^{\Mk^*} \dashv F^{\Mk^*}$, this is equivalent to
    \[ |\hom_{\EM^*}(\mathbb D, F^{\Mk^*}(R(\mathbb A)))| \,=\, |\hom_{\EM^*}(\mathbb D, F^{\Mk^*}(R(\mathbb B)))| \]
    for all $\mathbb D\in\EM^*$. Since $\EM^*$ is combinatorial by Corollary~\ref{cor:combinatorial-EM-f} and Remark~\ref{r:pointed-structures}, the last statement is equivalent to $F^{\Mk^*}(R(\mathbb A)) \cong F^{\Mk^*}(R(\mathbb B))$ and thus, by Lemma~\ref{l:rooting}, to $\overline{R}(F^{\Mk}(\mathbb A))\cong \overline{R}(F^{\Mk}(\mathbb B))$. Note that the restriction of $\overline{R}$ to the image of $F^{\Mk}$ is  full and faithful. Hence, $\overline{R}(F^{\Mk}(\mathbb A))\cong \overline{R}(F^{\Mk}(\mathbb B))$ if, and only if, $F^{\Mk}(\mathbb A) \cong F^{\Mk}(\mathbb B)$ which, by Proposition~\ref{p:Mk-logic}, is equivalent to $\mathbb A \equiv_{\CMLk} \mathbb B$.
\end{proof}

\section{Logical Normal Forms}\label{s:normal-forms}
Let $\CL^k_{\infty}$ denote the closure of $\CL^k$ under infinitary conjunctions and disjunctions.  This infinitary logic has been much studied in finite model theory as a means of delimiting the expressive power of fixed-point logics with counting.  In this logic we can express all, and only, the properties of finite structures that are invariant under the equivalence relation~$\equiv_{\CL^k}$.  In fact, since it is known that each finite structure is characterised up to $\equiv_{\CL^k}$ by a single sentence of $\CL^k$, it follows that every class of finite structures that is invariant under $\equiv_{\CL^k}$ is defined by a single infinitary disjunction of $\CL^k$ sentences (see~\cite{Otto1997} for details).  An entirely analogous normal form holds for $\CL_{n,\infty}$---the closure of $\CL_n$ under infinitary conjunctions and disjunctions (see~\cite[Chapter~8]{Lib2004}).

The theorems of Dvo{\v{r}}{\'a}k and Grohe provide a route to alternative normal forms for these infinitary logics.  To see this, note that for any finite structure $A$, we can write a primitive positive sentence $\gamma_A$ of first-order logic such that for any $B$, $B \models \gamma_A$ if, and only if, there is a homomorphism from $A$ to $B$.  Here, \emph{primitive positive} means that the sentence is built up from atomic formulas using only conjunctions and existential quantification.  The sentence $\gamma_A$ is called the \emph{canonical conjunctive query of $A$}  (see~\cite[Chapter~6]{Graedel2007FMT}).  It is known that if $A$ has tree-width strictly less than $k$, then $\gamma_A$ can be written using no more than $k$ variables~\cite[Lemma~5.2]{KolaitisVardi2000}.  Similarly, if $A$ has tree-depth at most $n$, then $\gamma_A$ can be chosen to have quantifier depth at most $n$~\cite[Lemma~2.14]{Rossman2008}.

Given a positive integer $t$, we can transform $\gamma_A$ by standard means into a sentence $\gamma_A^t$ of $\CL$ which asserts that there are at least $t$ distinct homomorphisms from $A$.  That is, for any $B$, $B \models \gamma^t_A$ if, and only if, $|\hom(A,B)| \geq  t$.  To show this, we establish something more general.

Let $\gamma(\bar{x})$ be a primitive positive formula with free variables among $\bar{x}$.
For any structure $B$ and any interpretation $\iota$ taking the variables in $\bar{x}$ to elements of $B$, we write $B \models \gamma[\iota]$ to indicate that the formula $\gamma$ is satisfied in $B$ when the free variables  $\bar{x}$ are interpreted according to $\iota$.  Since $\gamma$ is primitive positive, we have $B \models \gamma[\iota]$ precisely if there is a function $\kappa$ mapping the existential quantifiers in $\gamma$ to elements of $B$ so that every atomic formula $R(\bar{y})$ occurring in $\gamma$ is satisfied in $B$ when each variable $y$ in $\bar{y}$ is interpreted by $\iota(y)$, if this occurrence of $y$ is free in $\gamma$, and by $\kappa$ applied to the quantifier binding $y$ otherwise.  We write $|B \models \gamma[\iota]|$ to denote the number of distinct functions $\kappa$ witnessing  $B \models \gamma[\iota]$ and observe that $|B \models \gamma_A| = |\hom(A,B)|$.

We now define by induction on the structure of $\gamma$ a formula $\gamma^t(\bar{x})$ of $\CL$  with the property that for any structure $B$ and interpretation $\iota$ of $\bar{x}$ in $B$, we have $B \models \gamma^t[\iota]$ if, and only if, $|B \models \gamma[\iota]| \geq t$.  For atomic $\gamma$, $\gamma^1$ is just $\gamma$ and $\gamma^t$ is \texttt{false} if $t > 1$.  For $\gamma$  a conjunction $\gamma_1 \land \gamma_2$, let $T \coloneqq \{ (t_1,t_2) \in \N^2 \mid t_1\cdot t_2 \geq t\}$.  Then 
\[
\gamma^t \coloneqq  \bigvee_{(t_1,t_2) \in T} (\gamma^{t_1}_1 \land \gamma^{t_2}_2).
\]
When $\gamma$ is $\exists x. \gamma_1$, let $F$ be the collection of all finite partial functions $f$ on $\N$ such that $\sum_{s \in \mathrm{dom}(f)}sf(s) \geq t$.  Then
\[
\gamma^t \coloneqq \bigvee_{f \in F}\bigwedge_{s \in \mathrm{dom}(f)}\exists_{\geq f(s)}x. \gamma_1^s.
\]

Note that  both the quantifier rank and the total number of variables that appear in $\gamma^t_A$ are the same as for $\gamma_A$, and  $\gamma^t_A$ still has no negation symbols.  It then follows from Dvo{\v{r}}{\'a}k's theorem that any sentence $\phi$ of $\CL^k_{\infty}$ is equivalent in the finite (i.e., over finite structures) to an infinite Boolean combination of sentences of $\CL^k$ of the form~$\gamma^t_A$. Indeed, $\phi$ is equivalent to
\[ \bigvee_{B\in \mathcal M} \ \bigwedge_{A\in \mathcal{T}} \left( \gamma^{|\hom(A,B)|}_A
  \land \neg \gamma^{|\hom(A,B)|+1}_A\right) \]
where $\mathcal M$ consists of one representative for each $\equiv_{\CL^k}$-class of structures in $\{ B \mid B \models \phi\}$ and $\mathcal{T}$ consists of one representative for each isomorphism class of finite structures that have tree-width at most $k-1$.

Likewise, it is a consequence of Grohe's theorem that any sentence of $\CL_{n,\infty}$ is equivalent (in the finite) to an infinite Boolean combination of sentences of $\CL_n$ of the form $\gamma^t_A$; a similar normal form was also exhibited in \cite[\S 5]{grohe2020counting}.  This yields an interesting normal form for these counting logics.  In particular, there is no need of universal quantifiers or of equality, as neither of these is used in the sentences $\gamma^t_A$.  Further, if we allow dual counting quantifiers $\exists_{{\leq}i}$ with the semantics that $A \models \exists_{{\leq}i} x.\, \varphi$ if, and only if, $A \models \varphi[a/x]$ holds for at most $i$ distinct elements $a\in A$, then we do not need negation at all in our formulas.

From our point of view, the most interesting aspect of this is the elimination of equality symbols.  Consider for example the simple first-order sentence in the language of graphs: $\exists x\exists y.(x \neq y \land E(x,y))$ asserting the existence of an edge between two distinct vertices.  At first sight, it does not seem possible to express this property without using the equality symbol.  However, it is possible to do so with infinitary Boolean connectives and counting quantifiers as we illustrate.

First, consider the formula $\exists_{\geq t}xy. E(x,y)$ asserting that there exist at least $t$ pairs $x,y$ for which $E(x,y)$.  Note this is just $\gamma^t$ when $\gamma$ is the sentence $\exists x \exists y. E(x,y)$.   Now, the sentence $\exists x\exists y.(x \neq y \land E(x,y))$ is equivalent to
\[
\bigvee_{i \in \N} \left(\exists_{\geq i}xy. E(x,y) \land \exists_{\leq i-1} x. E(x,x) \right).
\]

We now give a direct proof, in our categorical framework, that we can eliminate equality from $\CL^k_{\infty}$ and $\CL_{n,\infty}$.

\begin{theorem}\label{th:equality-elim}
Every sentence of $\CL^k_{\infty}$ is equivalent in the finite to one without equality, and every sentence of $\CL_{n,\infty}$ is equivalent in the finite to one without equality.
\end{theorem}
\begin{proof}
We give the proof only for $\CL_{n,\infty}$ as the other is entirely analogous.  It suffices to prove that if two finite structures $A,B$ are not distinguished by any sentence of $\CL_n$ without equality, then they are not distinguished by any sentence of $\CL_n$.

Suppose then that $A$ and $B$ are not distinguished by any sentence of $\CL_n$ without equality.  We thus have $F^{\En}(A) \cong F^{\En}(B)$ in $\EM(\En)$ and so $\hom(D,F^{\En}(A)) \cong \hom(D, F^{\En}(B))$ for every $D$ in $\EM_f(\En)$.  By the adjunction $U^{\En} \dashv F^{\En}$, we have $\hom(U^{\En}(D),A) \cong \hom(U^{\En}(D),B)$ for every $D$ in $\EM_f(\En)$.  Since a finite structure has tree-depth at most $n$ if, and only if, it admits an $\En$-coalgebra structure, this implies $\hom(C,A) \cong \hom(C,B)$ for every finite structure $C$ with tree-depth at most $n$ and hence, by Theorem~\ref{thm:Grohe}, $A \equiv_{\CL_n} B$.
\end{proof}

\section{Conclusion}
The ideas developed in this paper connect homomorphism counting results in finite model theory with the framework of game comonads recently put forward by Abramsky et al.

Our categorical generalisation of Lov\'asz' theorem allows us to give uniform proofs of results by Dvo\v r\'ak and Grohe. The only part of the proof specific to each of these results is the combinatorial argument underlying the fact that the adjunction between $\sigma$-structures and $\sigma^+$-structures restricts to structures of bounded tree-width and bounded tree-depth, respectively. We also establish a new homomorphism counting result for modal logic which characterises the equivalence relation $\equiv_{\CMLk}$ over finite Kripke structures in terms of homomorphism counts.

As a by-product of the modularity of this general framework, we are immediately able to deduce: (i) the characterisation of the equivalence relation $\equiv_{\CL^k_n}$ in terms of homomorphism counts from structures admitting a $k$-pebble forest cover of height ${\leq\,}n$, (ii) the specialisation of our results to undirected graphs without loops (cf.\ Remark~\ref{rm:counting-tree-width-subclasses}), and (iii) equality elimination results for counting logics (Theorem~\ref{th:equality-elim}).

We believe that our method lays a pathway to discovering more Lov\'asz-type theorems. In particular, any comonad on the category of $\sigma$-structures that satisfies the conditions of Theorem~\ref{t:general-statement} will yield a Lov\'asz-type theorem. The natural next step to test this theory is to apply our results to the game comonads introduced in \cite{abramsky2020comonadic} and \cite{conghaile2021game}.

On the other hand, there are many Lov\'asz-type theorems that do not immediately fit into our framework. For example,  quantum isomorphism \cite{manvcinska2020quantum} and co-spectrality are characterised by restricting the homomorphism counts to planar graphs and cycles, respectively. Are there suitable comonads (or monads) that could allow us to recover these characterisations in a uniform manner? Another line for future investigation consists in fine-tuning the general theory introduced in this paper. For example, can we refine Theorem~\ref{thm:modal-counting} to yield a homomorphism counting result for p-morphisms, the natural concept of morphism between Kripke structures?

The ideas that we have described in this paper contribute to advancing the theory of game comonads, which connects two distinct areas of logic in computer science: on the one hand finite and algorithmic model theory, and on the other hand categorical and structural methods in semantics. Game comonads have been successfully applied, e.g., to investigate equirank-variable homomorphism preservation theorems in finite model theory~\cite{paine2020pebbling}. We broaden the applicability of this theory by bringing these techniques to the area of Lov\'asz-type results in finite model theory.

\section*{Acknowledgment}
 We acknowledge useful discussions in the Cambridge-Oxford Comonads seminar. The last-named author is grateful to Andr{\'e} Joyal for explaining to him (private communication) the construction $E(\, {\text -}\, )\mapsto \St{E}{\, {\text -} \, }$ for polyadic spaces~\cite{Joyal1971}, which features in Lemma~\ref{l:cores-isomorphic}.

%%%%%%%%%%   BIBLIOGRAPHY   %%%%%%%%%%

\bibliographystyle{plain}

\end{document}